\documentclass[
    amsmath,
    amssymb,
    amsthm,
    reprint,
    twocolumn,
    superscriptaddress,
    nofootinbib,
    longbibliography,
    aps,
    pra
    ]{revtex4-2}
    
\usepackage{graphicx}
\usepackage[colorlinks=true,linkcolor=blue,citecolor=blue,urlcolor=blue,plainpages=false,pdfpagelabels]{hyperref}
\usepackage{color}
\usepackage[dvipsnames]{xcolor}
\usepackage{mathtools}
\usepackage{yfonts}
\usepackage{float}
\usepackage{array}
\usepackage{verbatim}
\usepackage{physics}

\usepackage[utf8]{inputenc}
\usepackage[T1]{fontenc}
\usepackage{etoolbox}

\usepackage{anyfontsize}        
\usepackage{microtype}          
\usepackage[export]{adjustbox}  

\usepackage{newtxtext}
\usepackage{newtxmath}

\usepackage{dsfont}         
\usepackage{braket}         
\usepackage{soul}           
\usepackage{MnSymbol}       

\usepackage{multirow}                   
\usepackage{tabularx}                   
\usepackage[noEnd=false]{algpseudocodex}

\makeatletter
\g@addto@macro\bfseries{\boldmath}
\makeatother

\usepackage{silence}
\theoremstyle{definition}
\newtheorem{definition}{Definition}

\newtheorem*{rep@theorem}{\rep@title}
\newcommand{\newreptheorem}[2]{%
\newenvironment{rep#1}[1]{%
 \def\rep@title{#2 \ref{##1}}%
 \begin{rep@theorem}}%
 {\end{rep@theorem}}}
 \makeatother

\newtheorem{theorem}{Theorem}

\newtheorem{lemma}{Lemma}
\newreptheorem{theorem}{Theorem}

\newcommand{\RR}{\mathds{R}}

\usepackage{twemojis}

\definecolor{checkgreen}{rgb}{0.3, 0.8, 0.3}

\definecolor{warnred}{rgb}{0.8, 0.15, 0.1}

\definecolor{todored}{rgb}{1, 0, 0}

\definecolor{sorange}{rgb}{0.8, 0.3, 0}

\definecolor{forestgreen}{rgb}{0.0, 0.7, 0.0}

\definecolor{darkred}{rgb}{0.8, 0.0, 0.0}

\usepackage{tikz-network}
\usetikzlibrary {graphs}

\usepackage{enumitem}
\usepackage{cleveref}
\usepackage{orcidlink}
\newcommand{\MYhref}[3][blue]{\href{#2}{\color{#1}{#3}}}%
\begin{document}

\title{Full classification of Pauli Lie algebras}

\author{{Gerard~Aguilar}}
\affiliation{Dahlem Center for Complex Quantum Systems, Freie Universit\"{a}t Berlin, 14195 Berlin, Germany}

 \author{\MYhref[black]{https://orcid.org/0000-0002-9409-193X}{Simon Cichy}}
 \affiliation{Dahlem Center for Complex Quantum Systems, Freie Universit\"{a}t Berlin, 14195 Berlin, Germany}

\author{\MYhref[black]{https://orcid.org/0000-0003-3033-1292}{Jens~Eisert}}
 \affiliation{Dahlem Center for Complex Quantum Systems, Freie Universit\"{a}t Berlin, 14195 Berlin, Germany}

 \author{\MYhref[black]{https://orcid.org/0000-0003-1626-2761}{Lennart Bittel}}
 \affiliation{Dahlem Center for Complex Quantum Systems, Freie Universit\"{a}t Berlin, 14195 Berlin, Germany}

\date{\today}

\begin{abstract}
    Lie groups, and therefore Lie algebras, are fundamental structures in quantum physics that determine the space of possible trajectories of evolving systems. However, classification and characterization methods for these structures are often impractical for larger systems. In this work, we provide a comprehensive classification of Lie algebras generated by an arbitrary set of Pauli operators, from which an efficient method to characterize them follows. By mapping the problem to a graph setting, we identify a reduced set of equivalence classes: the free-fermionic Lie algebra, the set of all anti-symmetric Paulis on $n$ qubits, the Lie algebra of symplectic Paulis on $n$ qubits, and the space of all Pauli operators on $n$ qubits, as well as controlled versions thereof. Moreover, out of these, we distinguish 6 Clifford inequivalent cases and find a simple set of canonical operators for each, which allow us to give a physical interpretation of the dynamics of each class.
    Our findings reveal a no-go result for the existence of small Lie algebras beyond the free-fermionic case in the Pauli setting and offer efficiently computable criteria for universality and extendibility of gate sets. These results bear significant impact in ideas in a number of fields like quantum control, quantum machine learning, or classical simulation of quantum circuits.
\end{abstract}
\maketitle

\section{Introduction}

Lie groups emerge as a fundamental mathematical structure in quantum physics due to their intrinsic connection to the Schrödinger equation. The set 
of possible trajectories of an evolving quantum system is confined to the Lie group associated with the generators of the dynamics. Consequently, the study and classification of continuous symmetries and dynamics in quantum physics can be framed in terms of the analysis of such objects. Often, this analysis can be simplified by instead examining the tangent spaces of these Lie groups, namely Lie algebras, where the commutation relations of their elements encode the essential information about the evolution. 

This connection has been studied extensively in several fields within quantum information theory and quantum computation. In quantum control theory, universality and controllability questions are entirely ones of reachability, and hence Lie-algebraic in nature, making the study of these structures crucial for the design of universal quantum gate-sets and simulators \cite{QuantumControlReview,Dirr,PhysRevA.51.1015,PhysRevLett.74.4087}. Lie algebras have also been shown to play a crucial role in understanding and designing variational quantum algorithms \cite{Variational,McClean_2016} and quantum machine learning methods \cite{biamonte2017quantum}, where in certain scenarios have been directly linked to phenomena like barren plateaus and overparametrization \cite{BarrenPlateaus}. Additionally, approaches such as geometric quantum machine learning \cite{PRXQuantum.5.020328,PRXQuantum.4.010328,Anschuetz} also rely on Lie algebraic notions in order to simplify training by embedding problem specific symmetries into the circuits. Finally, ideas in quantum circuit complexity, like measures based on operators spread, are also strictly related to the Lie algebra spanned by the circuit generators \cite{Nielsen:2006mn2,ComplexityGrowth,PRXQuantum.2.030316,OperatorSpread} , and hence also the classical simulability of some families of quantum circuits \cite{LieSimulation,begusicSimulatingQuantumCircuit2023,ZoltanOld,gu2021fast}. 

Classification and characterization of Lie algebras are thus crucial for the further development of these fields. In a general setting however, this task often becomes intractable. Consequently, progress in this regard has been restricted to specific instances, where it is possible to rely on certain algebraic structural properties, like \textit{simplicity} \cite{cartan_simple}, or to construct an explicit basis of the Lie algebra, which is only feasible under strict assumptions on the generators \cite{wiersema2023classification,pozzoli2022lie,PhysRevLett.119.220502,chapman2020characterization,chapman2023unified,fendley2023free,bruschi2023deciding}, like locality or periodicity.

This work adds to this literature by giving a comprehensive classification of all \emph{ Pauli Lie algebras}, these being all Lie algebras generated by any set of Pauli operators, and also providing an efficient method to determine the Lie algebra of an arbitrary set of Pauli generators. Pauli operators are particularly interesting because they form a basis for all Hermitian operators. As such, they appear as natural generators in many popular gate sets in the context of digital quantum computing and become relevant in a plethora of tasks, like whenever one considers Hamiltonian simulation techniques based on product formulas \cite{lloyd1996universal,Childs2019fasterquantum,Fungible,Faehrmann,begusicSimulatingQuantumCircuit2023}.  Moreover, their particular commutation structure resembles that of Majorana operators in fermionic systems, thus making them directly useful for their study.

Due to these commutation and anti-commutation rules of Pauli operators, which we will leverage to simplify the problem, our work bears big resemblance to studies on (quasi-)Clifford Lie algebras in the mathematics literature \cite{gintz2018classifying, Cliffordalg, cuypers2021quasi}. 
While these studies focus on the mathematical structure of the Lie algebra, we aim to characterize these sets in a way that allows us to make further remarks about the embedding Hilbert space and the physics of their evolution.
The key component in our analysis will be to map the classification problem to a graph reduction problem. 
This approach will be similar in spirit to previous works like Bouchet's on graph equivalence under local complementations \cite{Bouchet} -- that has important implications
on assessing whether two graph states
\cite{Hein04}
are equivalent under sequences of local Clifford
conjugations \cite{PhysRevA.70.034302,Hein04} -- or others considering the vertex minor problem, where vertex deletion is also considered \cite{dahlberg2022complexity,PhysRevA.106.L010401}. In our case, we will find a set of universal equivalence classes for graphs under conditioned local complementations. This will allow us to identify the existence of 6 total Clifford inequivalent Lie algebras. Having that any given collection of initial operators can be mapped through Clifford operations to one of these six simpler sets, we will then derive their Lie algebra type and provide a physical interpretation of their structure along with their relationships.

This reveals that the biggest polynomially-sized Pauli Lie algebra is that of free-fermions together with their parity operator. All other non-free-fermionic cases correspond to the full set of Paulis on $n$-qubits as well as an embedding thereof in an $(n+1)$-qubit Hilbert space, the space of all anti-symmetric Pauli operators, and the space of symplectic Pauli operators, each exhibiting a dimension growing exponentially with system size. Beyond these cases, only controlled versions of these dynamics exist.

The remainder of this work is structured as follows. In \Cref{sec:definitions} we define the structures we are interested 
in, as well as the operations required for the reduction. In \Cref{sec:canonical_graphs_algebras} we introduce the main theorems of this work regarding graph reductions and their mapping to Lie algebra types. We add the distinction between Clifford inequivalent Lie algebras in \Cref{sec:clifford_equivalence}. With all Lie algebra types layed down, we discuss their physical interpretation in \Cref{sec:lie_algebra_description}. In \Cref{sec:disconnected_lie}, 
we make a final remark regarding disconnected graphs for completeness, and we finalize in \Cref{sec:conclusion} with a brief discussion of the implications of our results. We prove most of the statements in this work in the Supplemental Material.

\section{Definitions and setting}
\label{sec:definitions}
The statement of our results will require some preparation. Let us first define our main object of interest.

    \begin{definition}[Pauli Lie algebra]
        \label{def:Pauli_Lie}
        We define a Pauli Lie algebra $\mathfrak{g}$ as the Lie algebra whose elements are Pauli strings of some $n$-qubit system, which is closed under the Lie bracket, ($[\cdot,\cdot]_L:\mathfrak{g}\times\mathfrak{g}\rightarrow\mathfrak{g}$), given by the imaginary commutator of its elements, $\comm{P_i}{P_j}_L\coloneqq i\comm{P_i}{P_j} =  i(P_iP_j-P_jP_i)$.
        Then, a set of Paulis $\mathcal G=\{P_1,\dots,P_{n_G}\}$ is a generating set of $\mathfrak{g}$ if they span all elements in $\mathfrak{g}$ through Lie algebraic operations. Then we write $ \langle \mathcal G\rangle_{\RR,[\cdot,\cdot]}=\mathfrak{g}$.   A generating set is minimal, if removing any $P_i$ would change the size of the Lie-algebra. 
    \end{definition}

    Pauli Lie algebras have a very appealing structure since Pauli strings have simple commutation relations, namely given any two Paulis they either commute or anti-commute and the resulting operator is again a Pauli up to a constant. For two non-commuting Paulis $P,Q$, it holds that  
    \begin{align}
        [P,Q]= 2  P Q
    \end{align}
    where $P Q$ is itself another Pauli. Moreover, Pauli operators are particularly interesting since they form an orthonormal basis for the space of operators.
    A natural question to ask then, is whether one can characterize all Pauli Lie algebras generated by some arbitrary set $\mathcal{G}=\{P_1,\dots, P_{n_G}\}$.
     Like we already mentioned in the introduction, one common approach to address this, has been to explicitly construct a basis for the Lie algebra by repeatedly taking nested commutators of the generating elements \cite{wiersema2023classification,gohLiealgebraicClassicalSimulations2023}. However, since the size of the Lie algebra could be exponential, this method is in many cases not scalable. 

    Here, we will propose a new strategy that overcomes this by mapping the problem to a \emph{graph problem}. Exploiting some of the already mentioned properties of Pauli operators, we will be able to characterize Pauli Lie algebras for an arbitrary set of generators. This connection is expected to be
    fruitful, in a way that is reminiscent of 
    the
    way graph states \cite{Hein04} 
    constitute a theoretical
    laboratory to study questions in entanglement theory.

    Since the basis of the Lie algebra is generated through the commutator of its elements, and the commutator of two Paulis is either 0 or their product, all the information of the basis elements is already encoded exclusively in the commutation relations of the elements in $\mathcal{G}$, so long as the generators in $\mathcal{G}$ are Lie algebraically 
    independent, that is, no generator can be computed through commutators of others. This is satisfied if the generators form a minimal set. For now we assume that we are given a minimal set of generators, but then we lift this assumption in \Cref{sec:minimal_generators}.
    This allows us to study the Lie algebra by the anti-commutation relations of the generators $\mathcal{G}$.
    
    \begin{definition}[Anti-commutation graph]
        Given a set of Pauli operators $\mathcal{G}=\{P_i\}_{i=1}^{n_G}$, we call its anti-commutation graph the graph $\Gamma=(\mathcal{V},\mathcal{E})$, where every vertex corresponds to some $P_i$, and two vertices are connected by an edge if their corresponding Paulis anti-commute, i.e.,
        \begin{equation}
        \mathcal{V} = \{P_i\}_{i=1}^{n_G}, \; \mathcal{E} = \{(P_i,P_j)|[P_i,P_j]\neq 0\}.
    \end{equation}
    \end{definition}
    In general, this can be any possible graph for sufficently many qubits.
    In 
    \Cref{fig:canonical_types}(\textit{a}, left),
    we show the anti-commutation graph for the Lie algebraically independent set $\mathcal{G}=\{Z_0,X_0,Y_1X_2,X_1X_3,Z_3,Z_4X_3Z_1,X_4\}$.

   \begin{figure*}
       \centering
       \includegraphics[scale=0.21]{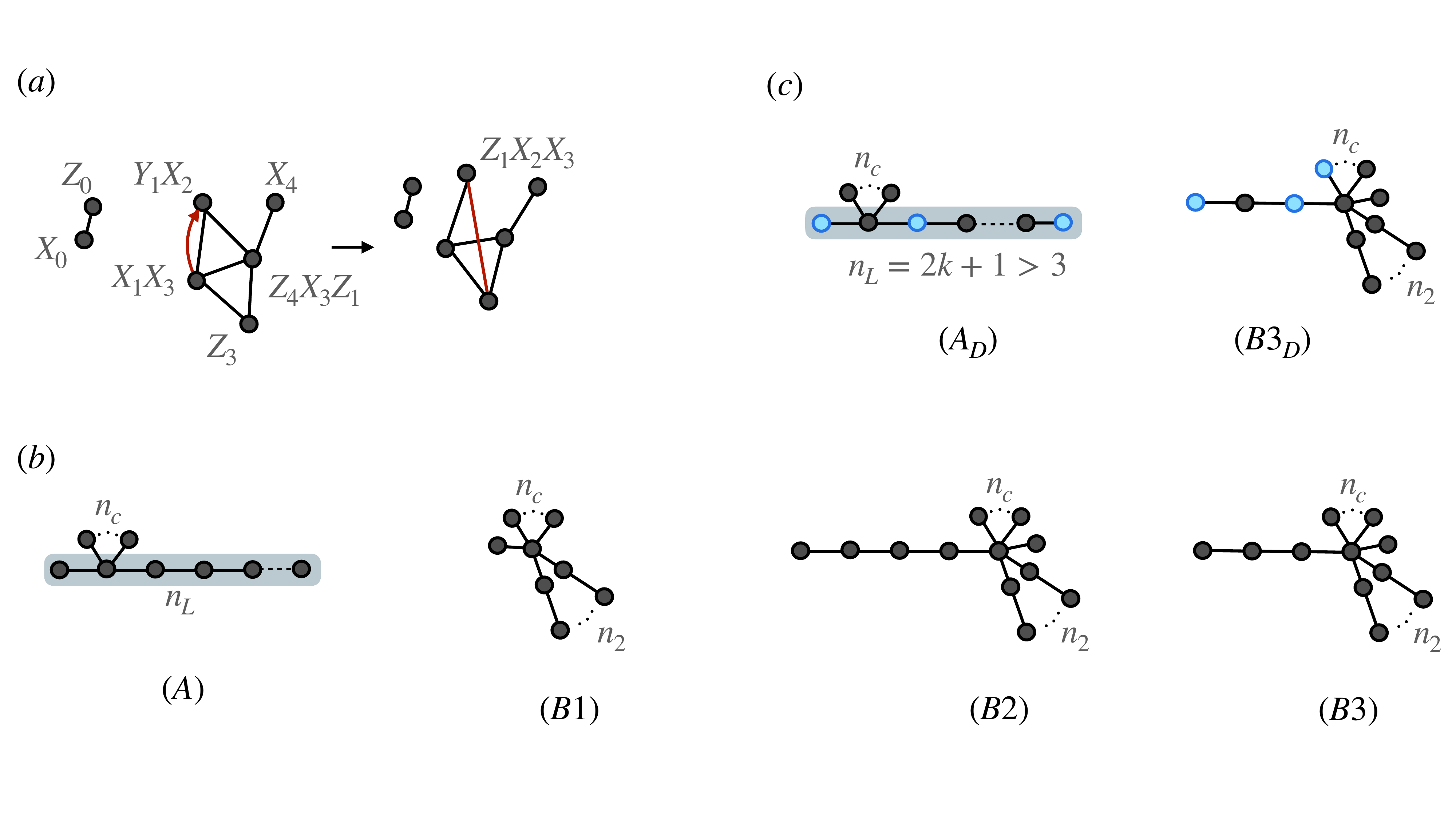}
       \caption{(a) Anti-commutation graph for the set 
        $\mathcal{G}=\{Z_0 ,X_0 ,Y_1X_2, X_1X_3, Z_3, Z_1X_3Z_4, X_4\}$. Contracting $X_1X_3$ onto $Y_1X_2$ substitutes the latter with $i\comm{X_1X_3}{Y_1X_2} \propto Z_1X_2X_3$, which modifies the connectivity as discussed in the main text. (b)
        Canonical graph representatives of all 4 universal equivalence classes under contractions as given in \cref{thm:types}.
        (c) Different algebraic dependencies render isomorphic Lie algebras inequivalent under unitary transformations. As discussed in \cref{thm:clif_algdeb} these only occur in graphs of type $\textit{A}$ and $\textit{B3}$ involving the blue coloured vertices, thus giving rise to 6 total unitary-inequivalent Pauli Lie algebras.
        }
        \label{fig:canonical_types}
   \end{figure*}
    
    Our goal will be to find a set of universal equivalence classes for graphs, that allows us to determine the Lie algebra associated to any given graph. To that end, we will define graph operations that leave the underlying Lie algebra invariant, and restate our problem as a graph reducibility problem under such operations. The operation in question will be a conditional local complementation. Namely, given two vertices $(p,q)$, we will complement edges $(p,v)$ conditioned on $q$ sharing an edge with $v$ or not. On the level of the Paulis, this will correspond to replacing, in the generating set, $\{P,Q\}\mapsto\{\pm i PQ,Q\}$, i.e., replacing a generator by its commutator with another generator. 
    
    \begin{definition}[Contraction]
        Given two non-commuting generators $P,Q\in \mathcal{G}$, a contraction of $Q$ onto $P$ maps $P\mapsto P'=\pm\frac{1}{2}i \comm{P}{Q}$ and $\mathcal{G}\mapsto\mathcal{G}'\ni P'$. This leaves the Lie algebra invariant $\mathfrak{g}=\mathfrak{g}'$.
    \end{definition}
    It is easy to see that indeed this operation does not change the Lie algebra because of the fact that $P=\pm \frac{1}{2} i\comm{P'}{Q}$. As such, the set generated from nested commutators of $\mathcal{G}'$ contains all those generated by $\mathcal{G}$.
     Like we said however, for the anti-commutation graph $\Gamma$, this transformation will result in a new graph $\Gamma'$ with modified connectivity. In more concrete terms, for the associated vertices $p,q$, corresponding to $P,Q\in\mathcal{G}$, a contraction of $q$ onto $p$ amounts to a complementation of the edges $E_p=\{(p,v) | (p,v)\in E \land (q,v)\in E\} $. 
     In terms of the adjacency matrix of the graph $A_{\Gamma}$, this is performed by adding the the columns and rows corresponding to vertex $q$ onto the columns and rows corresponding to vertex $p$ in $A_{\Gamma}$.
    
    An example of such a contraction is shown in 
    \Cref{fig:canonical_types}\textit{(a, right)}.
    In the following section we will describe universal equivalence classes under these contractions and how can that help us determining the Lie algebra in the most general case.
    
    \section{Canonical graph types and Lie algebras}
    \label{sec:canonical_graphs_algebras}
    
    We establish that using graph contractions, every anti-commutation graph is reducible to one of a few canonical types.
     \begin{theorem}[Canonical graphs]
        \label{thm:types}
    For any given connected graph $\Gamma$ with at least two vertices, there exists a sequence of contractions which result in one of four graphs:
    \begin{enumerate}[label=(\Alph*)]
        \item A line graph with with $n_L$ vertices as well as $n_c\in\mathbb{N}_0$ single vertices connected to the second to last vertex.
        \item A star graph with legs of length of at most $4$ and where the number of legs of length 1 and 2 are $n_c+1$, with $n_c\in\mathbb{N}_0$ and $n_2\in\mathbb{N}_+$, respectively. Additionally, they can present the following:
        \begin{enumerate}[label=(\arabic*)]
            \item no legs of length $3$ or $4$,
            \item one leg of length $4$ and no leg of length $3$,
            \item one leg of length $3$ and no leg of length $4$.
        \end{enumerate}
    \end{enumerate}
    The canonical 
    types are shown in \Cref{fig:canonical_types}(b).
    \end{theorem}
    \begin{proof}(Sketch)
        The full proof is in  \Cref{app:graph_equivalences}.
        The statement is proven by induction. The procedure works by selecting some initial vertex at random and iteratively adding new vertices. 
        We show that if the starting graph is one of the considered here, if one adds a new vertex with some arbitrary connectivity to it, there is always a sequence of contractions that transforms the new graph to one of these types.
    \end{proof}
    One important remark is that for any given initial graph $\Gamma$, the set of contractions that map it to one of these canonical representatives can be found efficiently. Since the proof works inductively in the number of vertices, applying a known set of transformations, one can find an algorithm that is efficient in the number of vertices of $\Gamma$ that finds the corresponding class. We leave the explicit algorithm for future work.

    With this, one can then classify the Lie algebras generated by a minimal set of Paulis by only looking at their graph and the canonical representative it maps to.
    This brings us to the main result providing  the  classification of the Lie algebra types.

    \begin{theorem}[Classification of Pauli Lie algebra types]
        \label{thm:graph_to_lie}
        Given a set of Lie algebraically independent Paulis $\mathcal{G}$, with connected anti-commutation graph $\Gamma$, then its Pauli Lie algebra corresponds to
        one of the following cases.
        \begin{itemize}
            \item $\bigoplus_{i=1}^{2^{n_c}}\mathfrak{so}(n_L+1)$   when $\Gamma$ maps to some $\Bar{\Gamma}$ in class \textit{A}.
            \item  $\bigoplus_{i=1}^{2^{n_c}}\mathfrak{sp}(2^{n_2})$   when $\Gamma$ maps to some $\Bar{\Gamma}$ in \textit{B1}.
            \item  $\bigoplus_{i=1}^{2^{n_c}}\mathfrak{so}(2^{n_2+3})$   when $\Gamma$ maps to some $\Bar{\Gamma}$ in \textit{B2}.
            \item  $\bigoplus_{i=1}^{2^{n_c}}\mathfrak{su}(2^{n_2+2})$   when $\Gamma$ maps to some $\Bar{\Gamma}$ in \textit{B3}.
        \end{itemize}
    \end{theorem}

    \begin{proof}
        (Sketch) The full proof can be found in the appendices. In \Cref{app:controls} we prove the $\bigoplus_{i=1}^{2^{n_c}}$ terms 
        .The main idea will be to show that several legs of length 1 act as a control register, thus splitting $\mathfrak{g}$ into two equal, smaller Lie algebras. The rest of the statement is proven by construction in \Cref{app:graph_lie_proof}. The key will be to come up with a \textit{good} set of canonical Pauli operators, for which it is easy to show, and since the Lie algebra type is uniquely determined by the graph, argue for the general case. For graphs in class \textit{A}, we will explicitly construct a set of free-fermionic majorana operators and show that this is indeed equivalent to $\mathfrak{so}(n_L+1).$ For graphs in \textit{B1} and \textit{B2}, we will show that some set of canonical generators obey the properties defining $\mathfrak{sp},\mathfrak{so}$ respectively, and thus they span a sub-algebra of those. On the other hand, we will come up with a set of rules to compute non-zero commutators on these graphs, and show that every element in $\mathfrak{sp},\mathfrak{so}$ can actually be written as such a commutator, thus showing equality with the Pauli Lie algebra spanned by our set. Finally, for graphs in \textit{B3} we will just show that every Pauli can be computed through nested commutators of our generators.
    \end{proof}
    Crucially, the classification of the Lie algebras does not depend directly on the number of qubits of the system. For our purposes, however, besides the structure of the algebra, it will also be relevant to see how a Lie algebra acts within the physical Hilbert space, which we will address in the following sections.
    
    \section{Clifford equivalence of Pauli Lie algebras}
    \label{sec:clifford_equivalence}
    So far we have discussed the anti-commutation relations for a minimal set of Pauli generators. However, Paulis can also have algebraic dependencies, meaning that even if no generator can be written as a non-zero commutator of others, we could have some generators being a product of others. This is relevant for our next theorem
    \begin{lemma}(Clifford equivalence lemma)
        Two Pauli Lie algebras are Clifford equivalent if they can be mapped to the same anti-commutation graph $\Gamma$, where the generators share the same algebraic dependencies.
    \end{lemma}
   
    \begin{proof}
        This follows from the general properties of Cliffords \cite{gottesman1998heisenberg}.
    \end{proof}
    In general it also holds that if two Lie algebras are not Clifford equivalent, then they are also not unitarily equivalent since no unitary can change commutation relations or algebraic dependencies.
    As such, we would like to determine all possible algebraic dependencies within the Pauli Lie algebra of some minimal generating set. We can establish the following theorem.
    
    \begin{theorem}[Limits to algebraic dependencies]\label{thm:clif_algdeb}
        Every minimal generator Pauli Lie algebra has at most one algebraic dependence. Algebraic dependencies can only occur for graphs of type \textit{A} and \textit{B3}. These can be mapped by contraction to only one possible case each which are shown in \Cref{fig:canonical_types} (c).
        As such, there is a total of $6$ Clifford inequivalent families of Pauli 
        Lie algebras.
    \end{theorem}
    
    \begin{proof}
        (Sketch) The full proof can be found in \Cref{sec:appendix_algd}, but we 
        state the main idea here.
        For this it is enough to check algebraic dependencies within the generating set of the graphs in \Cref{thm:types}, as contractions preserve algebraic and Lie algebraic dependencies. Given these specific anti-commutation relations, the sets of vertices that can be algebraically dependent are limited to subsets $\{v_i\}_{i=1}^m\subseteq\Gamma$ involving legs of odd length, where the $v_i$'s appear in an alternating fashion in $\Gamma$. Moreover, we will show that since algebraic dependence between legs of length 1 also implies Lie algebraic dependence, the only cases left are the ones stated in the theorem. 
    \end{proof} 
    \subsection{Finding the minimal set of generators}
    \label{sec:minimal_generators}
    The scenario described above assumes that we have a minimal set of generators. In general, this might not be the case however. If the set is not minimal, then a generator can be removed. The strategy is to still perform all the graph transformations to get to a canonical form which leaves the Lie algebra invariant even it is not the minimal set. By Gaussian elimination one can identify algebraic dependencies between generators, which is a necessary condition for Lie algebraic dependence. For the canonical types, algebraic dependencies can only occur between legs of length 1 and on a long leg of odd length in an alternating structure (as in \cref{thm:clif_algdeb}). Since the latter is an allowed dependence, a redundant generator will necessarily result in an algebraic relation involving only legs of length $1$. In this case one of the involved vertices can be safely removed without changing the Lie algebra. This process can be (efficiently) repeated until all Lie-algebraic dependencies are removed.
    \section{General description of all Pauli Lie-algebras}
    \label{sec:lie_algebra_description}
    In this section, we will give an intuition behind our classification result. First of all, we show in \Cref{app:controls} that the direct sums in \Cref{thm:graph_to_lie}, translate to an actual direct sum within the physical Hilbert space. Namely, given several legs of length $1$, using Cliffords, we can always map the Lie algebra to be of the form
    \begin{align}
    \label{eqn:controls}
    \mathfrak{g}=\mathrm{span}(\{Z_1^{i_1}\otimes\cdots \otimes Z_{n_c}^{i_{n_c}} \otimes P_{\mathfrak{g'}}\}|\vec i\in \{0,1\}^{n_c})
\end{align}
    where $\vec i\in \{0,1\}^{n_c}$ labels each block in the direct sum, and $P_{\mathfrak g}$ corresponds to operators in the Lie algebra where all additional legs of length 1 are removed ($n_c=0$). Then, from \Cref{eqn:controls} it is clear
    that these additional legs of length 1 always introduce a 
    new qubit register acting as a control on the dynamics given by $P_{\mathfrak g}$. 

    With this, we can now move on to considering each of the 6 types from \Cref{thm:clif_algdeb} without additional legs of length 1 (and hence qubit registers).
    From now on, we will refer to $n$ as being the minimal number of qubits required to span their dynamics (hence the total number of minimal qubits will actually correspond to $\Tilde{n}=n+n_C$). We will show in \Cref{app:controls} that this number $n$, like for the case of controls, can easily be derived from the number of vertices of the graph.
 
    For type $\textit{A}_I$, the number of qubits corresponds to $n=\left \lfloor n_L/2 \right \rfloor$. Then the Lie algebra corresponds to $\mathfrak{so}(2n)$ or $\mathfrak{so}(2n-1)$ depending on the parity of $|\mathcal{G}|$. In both cases, this corresponds to the free-fermionic Lie algebra, where $\mathfrak{so}(2n)$ can be seen as generated by all Majorana operators on $n$ modes, i.e., both first and second order operators
    \begin{align}
            \mathfrak g = \mathrm{span} \left( \{\gamma_i\gamma_j\}_{i< j \in [2n]}\cup\{\gamma_i\}_{i\in [2n]} \right)
    \end{align}
    while $\mathfrak{so}(2n-1)$ can be thought as being generated just by second order operators
    \begin{align}
            \mathfrak g=\mathrm{span}(\{\gamma_i\gamma_j\}_{i< j \in [2n]})
        .\end{align}
    For the algebraic dependent case $\textit{A}_D$, the generating set in terms of Majorana operators becomes $\mathcal{G}=\{\gamma_1,\dots,\gamma_{2n}, \gamma_1\cdots \gamma_{2n}\}$ which corresponds to the set of free-fermionic operations on $n$ modes, together with their parity operator have a Lie algebra of type $\mathfrak{so}(2n+1)$
    \begin{equation}
    \begin{aligned}
            \mathfrak g & = \mathrm{span} \bigg( \{\gamma_i\}_{i\in [2n]} \cup \{\gamma_i\gamma_j\}_{i< j \in [2n]} \\
            \quad & \cup \Big\{ \prod_{j\neq i} \gamma_j \Big\}_{i\in [2n]} \cup \Big\{ \prod_{j\in [2n]} \gamma_j \Big\} \bigg) 
    \end{aligned}
    \end{equation}
    which is then the largest connected Pauli Lie algebra on $n$ qubits of polynomial size. In this case, these new terms arising from having the parity operator as a generator lead to dynamics allowing more independent dynamics between both fermionic parity operators, compared to the \textit{$A_I$} case. A Lie algebra of this class is discussed in Ref.\ \cite{chapman2020characterization}, as arising from some many-body systems like the Kitaev honeycomb model in two dimensions \cite{Kitaev_2006}.
    
    Moving on to the exponentially big Lie algebras,  the number of qubits $n$ can be given as a function of the number of legs of length 2. Then, for $\textit{B1}$ the Lie algebra corresponds to $\mathfrak{sp}(2^n)$, which describes the space of all Paulis on $n$ qubits (here $n=n_2+1$) satisfying the symplectic condition, i.e., given some $Q:Q^T=-Q$
    \begin{align}
             \mathfrak g=\mathrm{span}(\{ P\in \mathcal{P}| P^T Q+QP=0 \})\,.
    \end{align}
    This Lie algebra has recently been motivated, and its properties have been further studied in Ref.\ \cite{garcia2024architectures}. The Lie algebra given by $\textit{B2}$ instead, corresponds to $\mathfrak{so}(2^n)$, i.e., the span of all purely imaginary Paulis ($P^T=-P$) in the same Pauli basis on $n$ qubits (here $n=n_2+2$). Then, for $Q^T=Q$, we have
    \begin{align}
             \mathfrak g=\mathrm{span}(\{ P\in \mathcal{P}| P^TQ+QP=0 \}).
    \end{align}

    This subalgebra has been previously studied in relation to Haar random circuits and benchmarking tasks, when one needs to restrict to some subspace of the full unitary group, for instance, when working with certain fault-tolerant error correcting codes \cite{hashagen2018real,harper2019fault}. Moreover, both $\mathfrak{sp}(2^n)$ and $\mathfrak{so}(2^n)$ can arise when considering operators in commonly studied 1-dimensional spin models, such as the Kitaev chain or the $XY$-model under certain interaction fields \cite{wiersema2023classification}.
    For the last two cases, for $\textit{B3}_D$ we have $\mathfrak{su}(2^n)$, the full Lie algebra of trace-less Hermitian matrices on $n$ qubits ($n=n_2+2$) 
     \begin{align}
            \mathfrak g=\mathrm{span}(\mathcal{P}\backslash \{\mathbf 1\})\, .
    \end{align}
    On the other hand, $\textit{B3}_I$ is equivalent to all Hermitian matrices on $n-1$ qubits $\mathfrak{su}(2^{n-1})$ ($n=n_2+3$), which can be understood as the embedding of a complex vector space into a real one ($\mathbb{C}^d\simeq \RR^{2d}$) \cite{halmos93}. The Paulis take the form $P\mapsto i\mathrm{Im}(P)\otimes \mathbf{1}+\mathrm{Re}(P)\otimes Y$. Then
    \begin{equation}
    \begin{aligned}
            \mathfrak g=\mathrm{span}(\{ P\in \mathcal{P}\backslash{Q}\,|\, [Q,P]=0, P^TK+KP=0\}),\\
            \mathrm{where} \, Q^TK+KQ=0\,.
    \end{aligned}
    \end{equation}
    Note that all the definitions given above for the exponentially big Lie algebras are given in a Clifford invariant form, i.e., under Clifford transformations we just have $Q\mapsto Q'$ preserving the (anti-) symmetric property.
   
    \section{Disconnected Pauli Lie-algebras}
     \label{sec:disconnected_lie}
    Hitherto the discussion centered around connected anti-commutation graphs. However, one could also have a generating set with a disconnected anti-commutation graph. In that case, in general, different connected components are Lie algebraically independent and, therefore, their Lie-algebras are direct sums (in terms of their representations) of the Lie algebras of the subgraphs
    \begin{align}
\mathfrak{g}\sim\bigoplus_{i=1}^{n_d}\mathfrak{g_i}
    \end{align}
    where $n_d$ is the number of connected subgraphs and $g_i$ is the Lie algebra of a connected subgraph.
    If there is no algebraic dependence between their generators, 
    then they are Clifford equivalent to Lie algebras acting on the tensor product of their Hilbert spaces.
    In general, generators of different sub Lie-algebras can still have an algebraic dependence between them, creating unitarily inequivalent Lie-algebras.  
    For this to occur, we need to consider the non-trivial center of a Lie algebra
    \begin{align}
        Z(\mathfrak{g})=\{X\in A(\mathfrak{g})| \forall Y\in \mathfrak{g}: [X,Y]_L=0\}
    \end{align}
    where $A(\mathfrak{g})\supset \mathfrak{g}$ is the algebra generated by the elements of $\mathfrak{g}$. To have an algebraic dependence, two sub Lie algebras need to share their non-trivial center 
    ($Z(\mathfrak{g_i})\cap Z(\mathfrak{g_j})\neq \mathrm{span}(\mathbf 1)$). The non-trivial center of the whole 
    Lie algebra we call $Z(\mathfrak{g}):=
    \mathrm{span}(\{Z(\mathfrak{g_i})\}_{i\in [d]})$. The non-trivial center arises from two sources:
    \begin{itemize}
        \item The product of two Paulis corresponding to vertices of legs of length 1.
        \item If the graph has a non control symmetry. This is the case if $\mathfrak{g}$ does have an algebraic dependence but does not, i.e., the graph is of a structure admissible to $\textit{A}_D$ or $\textit{B3}_D$, but no algebraic dependence.
    \end{itemize}
     As such, the number of independent generators of the algebra of the center is given by
    \begin{align}
        n_{Z}(\mathfrak{g}_i)= &n_c(\mathfrak{g}_i)+n_s(\mathfrak{g}_i)
    \end{align}
    where $n_s\in\{0,1\}$ describes the presence of the additional symmetry.
    With this the total number of generators in the non-trivial center is
    \begin{align}
        n_{Z}(\mathfrak{g})\leq \sum_{i=1}^{n_d}n_{Z}(\mathfrak{g_i})\eqcolon \bar n_{Z}(\mathfrak{g})\,,
    \end{align}
    where the difference between both quantities $n_Z(\mathfrak{g})$ and $\bar n_{Z}(\mathfrak{g})$ comes from algebraic dependencies between disconnected components.
    In that case, it suffices to describe the algebraic relation between a chosen set of generators of the respective centers of $Z(\mathfrak{g_i})$ and the generators of $Z(\mathfrak{g})$ ($\mathcal{B}=\{P_1,\dots ,P_{n_{Z} }\}$).
     The full set of dependencies can therefore be described by a matrix $M \in \{0,1\}^{n_{Z},\bar n_{Z}}$ where $Q_j$, corresponding to a generator of a non-trivial center of a sub Lie algebra, is given by $Q_j:=\prod_{i\in[n_Z]}P_i^{M_{i,j}}$. 
      For connected anti-commutation graphs, there are only polynomially many (in the number of vertices) Clifford inequivalent Lie algebras, corresponding to the $6$ canonical types. 
      If we allow for disconnected graphs, due to the freedom to reuse controls and potential symmetries in the case of $\textit{A}$ and $\textit{B3}$, the number of Clifford inequivalent Pauli Lie algebras becomes exponentially big in general.
    
    \section{Discussion and outlook}
    \label{sec:conclusion}
    
    In this work, we have presented a full classification of Pauli Lie algebras, as well as a sketch on how to algorithmically compute them. 
    Having found a universal set of equivalence classes for graphs under conditioned complementation, our results ultimately serve as a rigorous proof showing that every Pauli Lie algebra is, up to direct sums, either free-fermionic or exponentially large, including all types $\mathfrak{sp}(2^{n})$, $\mathfrak{so}(2^{n})$ and $\mathfrak{su}(2^{n})$. Additionally, we find an interpretation for direct sums in terms of the physical Hilbert space, as either tensor products, or as actual direct sums coming from gates acting as controlled operations. These considerations regarding the Hilbert space also reveal the existence of 6 Clifford inequivalent Lie algebras. In particular, we show that for $\mathfrak{so}(2^n)$ and $\mathfrak{sp}(2^n)$, 
    one can always find a Clifford mapping any generating set to another. In contrast, for the free-fermionic and $\mathfrak{su}(2^n)$, one can have two Clifford inequivalent versions. In the former case they are related through inclusion of the fermionic parity operator, and in the latter, by the embedding of the Lie algebra into a larger Hilbert space.
  
    We believe our results can be taken and interpreted in basically two ways. On the one hand, one can see them as a no-go result for the existence of sub-exponential Pauli Lie algebras other than the mentioned free-fermionic ones. 
    Free fermions been used as toy models and studied in a wide range of fields, including quantum machine learning and variational methods~\cite{bittel2021training,diaz2023showcasing}, as well as simulation techniques of quantum systems \cite{terhal2002classical,jozsa2008matchgates,mocherla2023extending}, and hence their behavior is well understood. This is relevant as our results limit any further avenues for methods within these fields relying on polynomially sized Pauli Lie algebras, like is the case for $\mathfrak{g}-sim$, geometric QML  and others.
   
    While our result only holds for Pauli generators, it would be interesting to study whether these restrictions hold for more general operators, e.g.,  when considering linear combinations of Paulis. In recent work, for instance, similar graph tools have been used to identify free-fermionic Hamiltonians given by arbitrary sums of Paulis \cite{chapman2020characterization, chapman2023unified}. It would be relevant to find out if this could be extended to identify all possible Lie algebras in a similar way as we do here.
   
    On a more positive note, our work can be viewed as a new framework for efficiently computing Lie algebraic criteria for universality and controllability of Pauli gate sets. Our techniques not only provide an answer to this, but also elucidate a clear way to extend some given operator pool to transition from one Lie algebra to another, offering valuable insights into the structure and capabilities of these gate sets. We believe that this can 
    be used to study properties of quantum circuits when doped with some additional set of gates outside their Lie algebra.
    
    To summarize, efficient and generic tools for classifying and characterizing Lie algebras are crucial in quantum information theory. This work then provides a comprehensive classification of Pauli Lie algebras without any additional assumption on the generators.  In doing so, we highlight the limitations that some quantum machine learning and simulation methods might face,  while also offering a new tool to address questions regarding reachability and equivalence of gate sets generated by Pauli operators. We believe this work can help in assessing questions that might have remained elusive to date, as well as motivate new algebraic methods in order to deal with other families of operators.
    
    \section*{Acknowledgments}
    We thank Lorenzo Leone, Greg White and Hakop Pashayan
 for useful discussions. 
This work has been supported by the BMBF (FermiQP, MuniQCAtoms, DAQC), the 
Munich Quantum Valley (K-8),
the Quantum Flagship (PasQuans2, Millenion),
the DFG (CRC 183), the Einstein Research Unit, 
Berlin Quantum, and the ERC (DebuQC).
 
\bibliography{bibliography.bib}

\onecolumngrid
\newpage
\appendix

\section{Preliminaries}
\label{app:preliminaries}
 In this section, we will  introduce some basic notions about the spaces and objects we work with.

 \begin{definition}[Lie algebra]
       We define a Lie algebra $\mathfrak{g}$ as a vector space over some field $K$ endowed with a  Lie product ($[\cdot,\cdot]_L:\mathfrak{g}\times\mathfrak{g}\rightarrow\mathfrak{g}$), under which it is closed, i.e., $\forall a,b\in \mathfrak{g}: [a,b]_L\in \mathfrak{g}$, such that the product satisfies
        \begin{itemize}
            \item Bilinearity.
            \item Anti-symmetry: $\comm{a}{b}_L=-\comm{b}{a}_L$.
            \item Jacobi identity: $\comm{a}{\comm{b}{c}_L}_L+\comm{b}{\comm{c}{a}_L}_L+\comm{c}{\comm{a}{b}_L}_L$=0.
        \end{itemize}
    \end{definition}

In our case we will focus on the space of $N \times N$ complex matrices  \( \text{Mat}_{N}(\mathbb{C}) \), where we will take the Lie product to be given by the complex commutator of two matrices, i.e., $\comm{A}{B}_L=i\comm{A}{B}=i(A\cdot B-B\cdot A)$. This Lie algebra is denoted by $\mathfrak{gl}(N,\mathbb{C})$. If one would remove the last two requirements from the previous definition, and simply impose bilinearity on the product, that would define an \emph{algebra} instead. In many cases, we will be interested in looking at products of elements in $\mathfrak{gl}(N,\mathbb{C})$, which might not be in the Lie algebra, but in their algebra.

In this work, we will have a particular focus on the 
\emph{Pauli group}, which is a subgroup of \( \text{Mat}_{2^n}(\mathbb{C}) \), defined as
\begin{equation}
    \mathcal{P}_n = \left \{ \left. e^{i\phi\pi/2}\sigma_{i_1}\otimes\dots\otimes\sigma_{i_n} \right| \phi\in\{0,1,2,3\},i_k\in\{0,1,2,3\} \right \}
\end{equation}
for 
\begin{equation}
    \sigma_0=\mathbf{1}=\begin{pmatrix}
        1 & 0\\
        0 & 1
    \end{pmatrix},
    \hspace{0.5cm}
    \sigma_1=X=\begin{pmatrix}
        0 & 1\\
        1 & 0
    \end{pmatrix},
    \hspace{0.5cm}
    \sigma_2=Y=\begin{pmatrix}
        0 & -i\\
        i & 0
    \end{pmatrix},
    \hspace{0.5cm}
    \sigma_3=Z=\begin{pmatrix}
        1 & 0\\
        0 & -1
    \end{pmatrix}.
    \hspace{0.5cm}
\end{equation}
Another particularly interesting group that will be discussed over this work is, in fact, the \emph{Clifford group}, which is the \textit{normalizer} of $\mathcal{P}_n$
given by
\begin{equation}
    \mathcal{C}_n = \left\{ C\in U(2^n)|CPC^{\dagger}\in\mathcal{P}_n \right\},
\end{equation}
i.e., the set of all unitary operations on $n$ qubits that 
leaves the Pauli group invariant. It will be interesting because given two elements of $\mathcal{P}_n$, their Lie product will be preserved under Clifford operations, and thus we will be able to transform between sets without changing the underlying Lie algebra.

There are many properties of $\mathcal{P}_n$ which make Paulis interesting, one being that they form a basis for the space of all operators on $n$ qubits. In particular, if one restricts to coefficients over $\mathbb{R}$, they span the set of all Hermitian operators. When one excludes the $\sigma_0^{\otimes n}$ term, this gives a basis for all (Hermitian) traceless operators. 
Pauli commutation and anti-commutation rules are particularly interesting, given that for two Paulis $P_1,P_2\in\mathcal{P}_n$, $P_1P_2=\pm P_2P_1$ and hence, they either commute or anti-commute. These properties will be instrumental for this work. 

It is easy to check that the hermicity and traceless properties are preserved under the Lie product and hence, these two conditions define an actual sub-algebra of $\mathfrak{gl}(N,\mathbb{C})$, namely
\begin{equation}
     \mathfrak{su}(N) \equiv \{P \in \text{Mat}_{N}(\mathbb{C}) \mid P^\dagger = P, \, \text{Tr}(P) = 0\}
\end{equation}
i.e., the Lie algebra spanned by all non-identity Paulis over real coefficients. If one would add $\mathbf{1}$ to this set then the Lie algebra is called $\mathfrak{u}(N)$. In our case, we will be interested in $\mathfrak{su}(N)$ and sub-algebras thereof. One such sub-algebra can be constructed by restricting to anti-symmetric matrices, which for Hermitian matrices (of which Pauli matrices is a subset of) is equivalent to purely imaginary ones
\begin{equation}
    \mathfrak{so}(N) \equiv \left\{P\in\mathfrak{su}(N)|P^T=-P \right\} \subset 
    \mathfrak{su}(N).
\end{equation}
On the other hand, for even $N$, one could restrict to the matrices in $\mathfrak{su}(N)$ satisfying the equation
\begin{equation}
    P^T\Omega=-\Omega P
\end{equation}
for $\Omega$ some non-degenerate, anti-symmetric bilinear form. Typically, however, $\Omega$ is defined as
\begin{equation}
    \Omega \equiv \begin{pmatrix}
        0 & \mathbf{1}_N \\
        -\mathbf{1}_N & 0
    \end{pmatrix} = iY_0\otimes\mathbf{1}_N
\end{equation}
which then gives a canonical definition of the symplectic Lie algebra
\begin{equation}
    \mathfrak{sp} \left( \frac{N}{2} \right) \equiv \left\{ P\in\mathfrak{su}(N)|P^TY_0=-Y_0 P \right\} \subset \mathfrak{su}(N).
\end{equation}
One particular construction of a Lie algebra which we will be interested in is when this is defined with respect to some generating set.
\begin{definition}[Generating set]
    We say that some set of Paulis $\mathcal{G}=\{P_1,\dots,P_{n_G}\}$ is a generating set of some Lie algebra $\mathfrak{g}$ if they span all elements in $\mathfrak{g}$ through Lie algebraic operations. Then we write $ \langle \mathcal G\rangle_{\RR,[\cdot,\cdot]}=\mathfrak{g}$. A generating set is minimal, if removing any $P_i$ would change the size of the Lie algebra. 
\end{definition}

When some Lie algebra $\mathfrak{g}$ is defined as being generated by some set $\mathcal{G}$, this is sometimes referred as the Dynamical Lie algebra of $\mathcal{G}$.

One of the main questions of interest regarding Lie algebras is their classification. When a Lie algebra $\mathfrak{g}$ is defined as being generated by a set $\mathcal{G}$ one would also be interested in methods that efficiently characterize these spaces. When dealing with Paulis, and hence sub-algebras of $\mathfrak{su}(2^n)$, certain results exist. In particular, when one considers sub-algebras of $\mathfrak{u}(2^n)$ (i.e., including $\sigma_0^{\otimes n}$), a classification in terms of direct sums of $\mathfrak{su},\mathfrak{sp},\mathfrak{so},$ plus some \textit{Abelian} component has been developed. These sub-algebras are usually called \textit{reductive}, and in particular \textit{semisimple} when the Abelian component is null. While this is a useful result that can help characterize some of these Lie algebras, it is in general practically unfeasible to find such a decomposition, given a set of generators, and thus the task remains difficult in spite of these sparse results.

In this work, we will sometimes refer to the center of a Lie algebra, $Z(\mathfrak{g})$ as the set of operators commuting with every element in $\mathfrak{g}$. If there is some Abelian component in the decomposition of $\mathfrak{g}$ it will be a subset of it. If we think of $\mathfrak{g}$ as a sub-algebra of some bigger space, we can also have some elements $z\in Z(\mathfrak{g})$ such that $z\notin \mathfrak{g}$, but potentially in the algebra.
We refer to Ref.\ \cite{knapp1996lie} for a more in depth discussion about these objects and their properties. 

Finally, a particular class of interesting systems that will repeatedly appear in this work are free-fermionic systems. These are many-body systems consisting of a set of non-interacting particles. In our context, they are particularly interesting for two reasons. Firstly, they can be described by a set of Hermitian operators with similar commutation structure to Pauli matrices, namely Majorana fermion modes $\{\gamma_i\}_{i\in[2n]}$, which satisfy

\begin{equation}
    \{\gamma_i,\gamma_j\}=2\delta_{i,j}\mathbf{1}.
\end{equation}
As such, there are several ways one can map these operators to Paulis, one of the most common ways being through the Jordan-Wigner transformation, which maps
\begin{equation}
    \gamma_{2j-1}=\Pi_{k=1}^{j-1}Z_k\otimes X_j, \hspace{1cm} \gamma_{2j}=\Pi_{k=1}^{j-1}Z_k\otimes Y_j.
\end{equation}
Secondly, they are known to be solvable by classical methods, and, in fact, are one of the only known systems to span a polynomially large Lie algebra. To see this, it can be shown that such free-fermionic Hamiltonians can be written as quadratic in the Majorana modes
\begin{equation}
    H=2i\sum_{i,j\in[2n]}h_{i,j}\gamma_i\gamma_j\eqqcolon i\gamma\cdot h\cdot \gamma^T
\end{equation}
with $\gamma$ a vector of the Majorana operators, and $h$ a real anti-symmetric coefficient matrix. Due to the canonical anti-commutation relations,
\begin{equation}
    \comm{\gamma\cdot h\cdot\gamma^T}{\gamma_i} = -4(h\cdot\gamma^T)_i
\end{equation}
for some individual mode $\gamma_i$. Hence, these evolve under $H$ as
\begin{equation}
    e^{iHt}\gamma_ie^{-iHt}=\sum_j (e^{4ht})_{i,j}\gamma_j.
\end{equation}
Since $h$ is anti-symmetric and real, $e^{4ht}\in SO(2n,\mathbb{R})$, and  there $\exists W\in SO(2n,\mathbb{R})$ s.t.
\begin{equation}
    W^T\cdot h\cdot W  = \bigotimes_{j=1}^n\begin{pmatrix}
        0 & \lambda_j\\
        -\lambda_j & 0
    \end{pmatrix}\,
\end{equation}
with real $\{\lambda_j\}$. The special orthogonal group
$SO(2n,\mathbb{R})$
preserves the parity of fermion number.
$W$ itself can be represented as the exponential of a quadratic Majorana fermion operator as well, so that $W\eqcolon e^{4\omega}$. Finally, the Hamiltonian $H$ can be solved by exact diagonalization as 
\begin{equation}
    e^{-\gamma\cdot w\cdot\gamma^T}He^{\gamma\cdot w\cdot\gamma^T}=2\sum_{j=1}^n\lambda_j\gamma_{2j-1}\gamma_{2j}=2\sum_{j=1}^n\lambda_j Z_j.
\end{equation}

\section{Graph equivalences }
\label{app:graph_equivalences}
In this section,  we prove \Cref{thm:types} and introduce some of the tools required for it. Since our goal is to classify Pauli Lie algebras in all generality, we cannot rely on constructive arguments like in Ref.\ 
\cite{wiersema2023classification}, where locality constraints make explicit computation of $\mathfrak{g}$ feasible. Instead, we map the problem to a graph reduction problem and find all equivalence classes under contraction operations as defined in the main text. Recall that contractions correspond to performing a shuffle of the set of generators on the Lie algebra level, which on the graph will act as a conditioned local complementation. In order to facilitate discussion of sequences of contractions on the graph we will introduce the following tool verbatim to Ref.\ \cite{gintz2018classifying}.

\begin{definition}[Lightning $\chi$] Given a graph $\Gamma=(\mathcal{V},\mathcal{E})$, and a vertex $V\in \mathcal{V}$ onto which we wish to perform some contractions, we define the lightning $\chi=\Gamma(V)$ with respect to V as the induced sub-graph $\Gamma\setminus \{V\}$ with labeled vertices: \textit{lit} (\textit{unlit}) if on $\Gamma$ they are connected to V (if on $\Gamma$ they are not connected to V). We then say we \textit{toggle} a lit vertex $\omega$ whenever we perform a contraction of $\omega$ onto $V$. Since this operation on $\Gamma$ complements edges $(n_{\omega},V)$, for every $ n_{\omega}$ in the neighbourhood of $\omega$, in terms of the lightning, it changes the \textit{lit / unlit} state of all $n_{\omega}$. 
\end{definition}

Graphically we will represent lit vertices as white vertices and unlit ones as black vertices.
We show in an example of a lightning in \Cref{fig:lightning}.
\begin{figure}[h]
    \centering
    \includegraphics[scale=0.15]{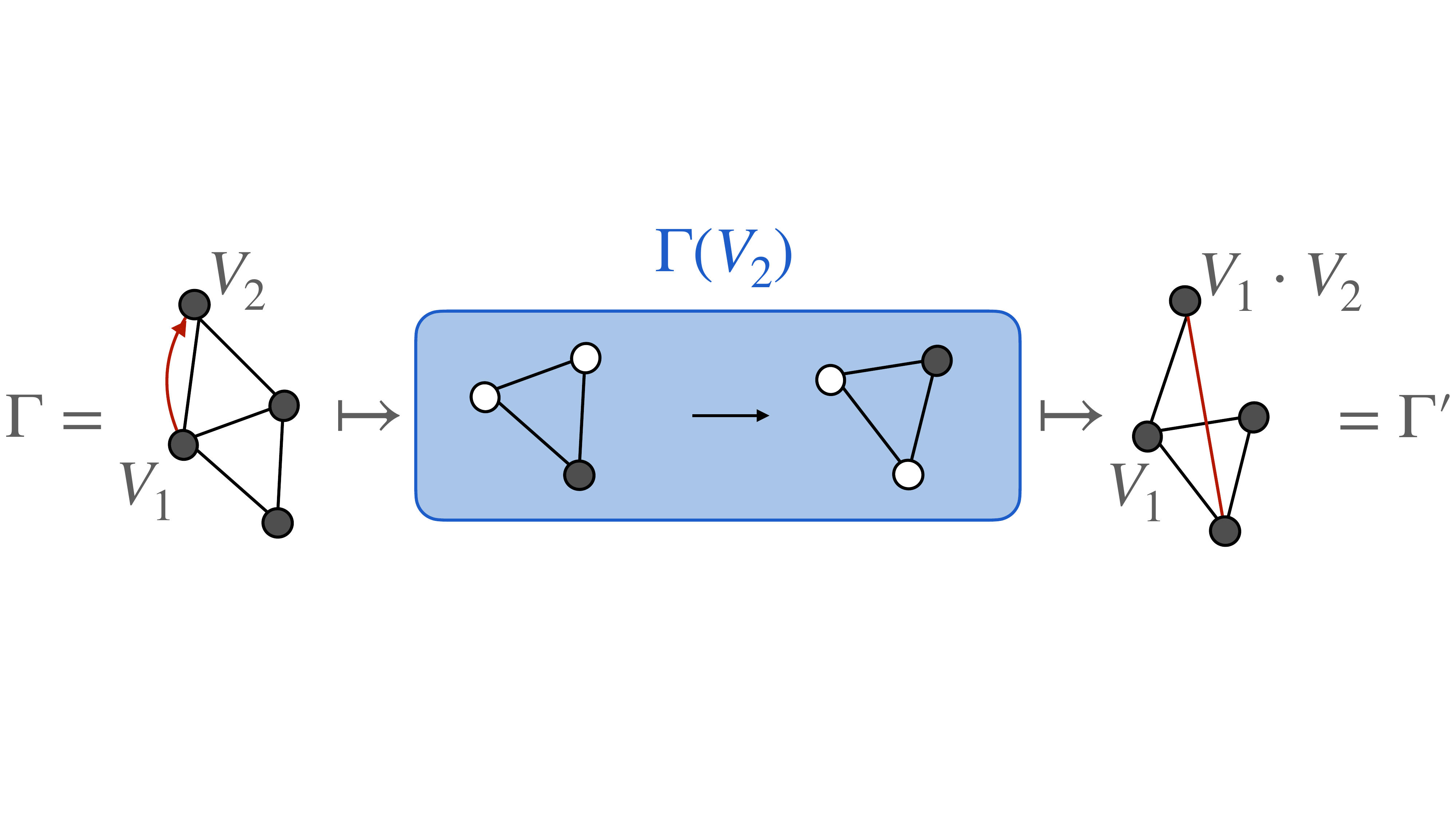}
    \caption{An example of how lightnings can be used to visualize contractions. Contracting $V_1$ onto $V_2$ can be seen as taking the lightning $\Gamma(V_2)$, and toggling $V_1$. This changes the lit / unlit satate of $V_1$'s neighbours, which changes the connectivity of $\Gamma$ as expected.}
    \label{fig:lightning}
\end{figure}

With this, our approach will be tightly related to the \emph{lit-only $\sigma$-game}, where given two lightnings on a graph, the goal is to determine equivalence through some adequate sequence of toggles \cite{goldwasser_does_2009,wang_minimum_2007,sjostrand2023orbits}. Big part of our result will then be based on finding equivalent lightnings for arbitrary initial ones on our set of canonical graphs.

From now on, we will say a vertex $V$ is central whenever $deg(V)\geq 3$. Typically, in our graphs we will only have one such vertex which we will call $O$. Then, we will refer to every connected component of $\Gamma\setminus\{O\}$ as a leg of $O$. Given a leg $L$, we will use $L_i$ for the vertex in $L$ at distance $i$ from $O$.

With this, we can now move on to proving our statement. To that end, we will introduce first two lemmas that will be instrumental.

\begin{lemma}
    \label{lemma:light_even_legs}
    Given a lightning on a star graph $\Gamma$ (i.e., only one central vertex) and at least one leg of length 1, where the central vertex $O$ is lit. Take two legs $L$ and $M$ of length $l_L\in\{2,3,4\}$ and $l_M=2$ respectively. Then, as long as $L_2$ and $M_2$ are initially unlit, there is a sequence of toggles that inverts the state of $L_1,M_1$ and $L_3$ if exists, while leaving the rest of the lightning invariant.
\end{lemma}

\begin{proof}
    Assume first that both $L_1$ and $M_1$ are in the same intial state, say for now that both are unlit. Then we just need to toggle 
    \begin{equation}
        O,M_1,L_1,O,M_{2\rightarrow 1},L_{2\rightarrow 1}.
    \end{equation}
    If they are both initially lit then toggling according to the same sequence in reverse order unlights them both. Besides this, all other vertices remain unchanged, except for $L_3$ if it exists since we toggled $L_2$ only once, and hence $L_3$ is flipped.

    On the other hand, $wlog$ we can assume $L_1$ is lit and $M_1$ is unlit. Then, for $\omega$ a leg of length 1, which for now we assume is lit, we toggle according to

    \begin{equation*}
        L_1,L_2,\omega,O,L_1,M_1,O,M_2,M_1,\omega.
    \end{equation*}
    After this, $L_1$ is unlit, $M_1$ is lit and since we toggled $O$ an even number of times, the rest of the lightning remains the same, except again for $L_3$ potentially. For the previous sequence, we assume that $\omega$ is initially lit. If it was not, we just need to toggle $O$ once before, apply the same sequence, and toggle $O$ once more at the end of it.
\end{proof}

With this we can prove the following lemma about equivalence between graphs with one single center and our canonical graphs.

\begin{lemma}
    
    \label{lemma:shortening}

    Given a star graph $\Gamma$ (i.e., only one central vertex) with at least one leg of length 1, and any number of legs of arbitrary length, they can always be transformed by contractions into one of our canonical graphs.
\end{lemma}

\begin{proof}

    Let us call $O$ the central vertex, and $\omega$ one leg of length 1 (we ask for $\Gamma$ to have at least one such leg). First of all, if the central vertex $O$ has $N_L$ legs, of which at least $N_L-1$ are of length 1, then $\Gamma$ is of type $A$.

    Otherwise, there are at least two legs of length $l\geq 2$. In that scenario, 
    if there is a leg $L$ of length $l>4$ we can transform it into a leg of length 4 and one of length $l-4$. To see this, take a lightning $\chi=\Gamma(L_5)$. Call $M$ the other leg of length bigger than 1, and $P$ any other neighbour of O (which must exists since O has degree at least 3). Then toggle 
    \begin{equation}
        L_4,L_3,L_2,L_1,O,M_1,P,O,L_1,M_2,M_1,O,L_2,L_1,L_3,L_2,L_4,L_3,P,O,L_1,L_2,M_1,M_2,O,L_1,M_1,O.
    \end{equation}
    This leaves only the $O$ vertex lit, i.e., we removed the edge between $L_4$ and $L_5$, and created an edge between O and $L_5$.
    Repeating this we reduce the maximum length of a leg in the graph to 4.

    Moreover, if $\Gamma$ has several legs of length 3, it can always be transformed into a graph with only one leg of length 3. Let $L,M$ be two such legs. Then take a lightning $\chi=\Gamma(L_3)$ and toggle according to
    \begin{equation}
    L_2,L_1,O,M_1,M_2,M_3,\omega
    \end{equation}
    applying \Cref{lemma:light_even_legs} to $L$ and $M$ lights $L_1,M_1$ and unlights $L_3$. Toggling $O$ once more unlights all vertices but itself. Hence, we managed to transform $L$ into a leg of length 2 and a leg of length 1. This procedure can be repeated until only one leg of length 3 is left.

    Finally, we will prove that legs of length 4 can be transformed into legs of length 2 under certain conditions. On the one hand, if our graph has no leg of length 3, then let $L$ and $M$ be two of the legs of length 4. Then take a lightning $\chi=\Gamma(L_3)$ and toggle according to 
    \begin{equation}
        L_2,L_1,O,M_1,M_2,M_3,M_4,\omega,O,L_1,L_2,M_1,O,\omega,L_1,M_4,M_3,O,M_1,M_2.
    \end{equation}
    This transformation results in $M_2$ having $deg(M_2)=3$ and two legs of length 2, which we call $P$ and $Q$. Now, through an adequate sequence of contractions we will remove these legs from $M_2$ and append them to $O$. In particular, we contract
    \begin{align}
            M_2 &\mapsto \comm{M_2}{M_1}, \\
            O &\mapsto \comm{O}{M_2}.
    \end{align}
    In particular, this creates an edge between $O$ and $P_1$, and $O$ and $Q_1$. We can now take a lightning $\chi=\Gamma(M_2)$, where $O,P_1$ and $Q_1$ are all lit. Applying \Cref{lemma:light_even_legs}, onto $P$ and $Q$, unlights them without changing the rest of the lightning. This just disconnected $P$ and $Q$ from $M_2$, and connects them to $O$. 
    At the end of these transformations, $L$ and $M$ have both been broken down into two legs of length 2 each. This transformation can be repeated for every pair of legs of length 4. After this transformation we will always have reduced our $\Gamma$ to a graph either in $B1$ or $B2$, depending on the parity of legs of length 4.
    On the other hand, if $\Gamma$ has one leg of length 3, then let $M$ be the leg of length 3 and $L$ a leg of length 4. We can take a lightning $\chi=\Gamma(L_3)$ and toggle according to 
    \begin{equation}
        L_2,L_1,O,M_1,M_2,M_3,\omega,O,L_1,L_2,M_1,O,\omega,L_1,M_3,\omega,M_3,M_2,M_1,O.
    \end{equation}
    This breaks down $L$ into two legs of length 2. Repeating this for every leg of length 4, results in a graph $B3$.
 \end{proof}

By making use of these two lemmas, we will be able to first prove that any lightning on one of our canonical graphs, is equivalent to a lightning where only one vertex is lit. From this it will follow that there will always be a sequence of contractions mapping back to one of the canonical graphs.

\begin{theorem}[Equivalence of lightnings]
    \label{thm:1_lit}
    Given a lightning $\chi$ on some graph $\Gamma$ of a canonical type, it is always equivalent to some lightning where only one vertex is lit.
\end{theorem}

\begin{proof}
    We will first make a comment regarding multiple legs of length 1.
    \paragraph*{Legs of length 1 in different initial lit states.}
    First, assume that not all legs of length 1 have the same lit configuration in $\chi$. Let $P$ be one of the lit legs, and $Q$ an unlit one. Then, given any generator $g\in\mathcal{G}$ other than $P$ or $Q$ we can find a sequence of contractions that maps $g\mapsto g'=PQ\cdot g$ and leaves the graph invariant (this is clearly the case since $PQ$ commutes with every element in the Lie algebra). If $g=O$ then we can contract $P$ and $Q$ onto O since both anti-commute with it. For $g=L_j$ any other vertex on a leg $L$, we contract onto $L_j$ the following sequence of vertices
    \begin{equation}
        L_{j-1},\dots,L_{1},O,P,Q,O,L_1,\dots,L_{j-1}.
    \end{equation}
    It is easy to check that all of these contractions are allowed, and transform $g$ so that the final generator is $g\cdot PQ$. Now, since $P$ was lit, the transformation $g\mapsto g\cdot PQ$ changes the lit state of $g$. Then, we can perform this operation freely to every lit generator $g\neq P$, so that after that, all vertices are unlit but $P$. 
    
    \paragraph*{Legs of length 1 in the same initial lit state.} If all legs of length 1 are in the same initial state, since this can only change by toggling $O$, which flips all legs, they will always be in the same state. Then, for graphs in class $A$ let us establish that legs of length 1 are on the left end of the graph, and the central vertex $O$ corresponds to the second leftmost vertex. Then we will proceed by always toggling the second leftmost vertex in $\Gamma$. Repeating this procedure we always end up moving the leftmost lit vertex one position to the right. Thus, ultimately all lit vertices concentrate in a smaller region of the path and hence eventually cancel all out but one. In case only the legs of length 1 are lit, we toggle one of them, in order to light $O$, and then proceed in the same way. For all other cases, i.e., graphs of type $B1,B2,B3$, we will show that we can easily take care of at least all but one legs of length 2, and map the discussion to a particular case of the above, namely a small path graph. In this scenario, we can assume $wlog$ that the center of the graph $O$ and all legs of length 1 are lit. If not, we simply need to take some other leg $L$ with at least one lit vertex, and toggle the innermost lit vertex until $O$ is lit. Toggling $O$ then lights the legs of length 1. 
    
    Moreover, for any leg $L$ of length 2, with at least one lit vertex, we can always assume only $L_1$ is lit. If both were lit, we can toggle $L_1$. This unlights $L_2$ and $O$, which we can light back by toggling $\omega$. If only $L_2$ was lit, we can toggle it and we are in the previous case. For the leg of length 3 or 4 in $B3$ and $B2$ respectively, the same logic follows. From this we can assume that the second vertex of every leg is always unlit, which in turn allows us to apply \Cref{lemma:light_even_legs} freely between any pair of legs. Hence, let us call $L$ either the leg of length 4 in $B2$, the leg of length 3 in $B3$ or some randomly chosen leg of length 2 in $B1$. For any unlit leg $M$ of length 2, applying \Cref{lemma:light_even_legs} lights up $M_1$, while inverting the state of $L_1$ and possibly $L_3$. Repeating this for every unlit leg of length 2 and toggling $O$, unlights all neighbours of $O$ except $L_1$ potentially. If after this step only $O$ is lit we are done. Otherwise, we just need to solve a small instance of the previous case for graphs in \textit{A}, where again, we will toggle the second innermost lit vertex in L repetitively. By doing so, we always end up with just one lit vertex in L.
\end{proof}

With this, we can now prove \Cref{thm:types}.

\begin{reptheorem}{thm:types}[Canonical graphs]
       
    For any given connected graph $\Gamma$, there exists a sequence of contractions which result in one of three graphs:
    \begin{enumerate}[label=(\Alph*)]
        \item A line graph with with $n_L$ vertices as well as $n_c\in\mathbb{N}_0$ single vertices connected to the 
        second to last vertex.
        \item A star graph with legs of length of at most $4$ and where $n_c,n_2\in\mathbb{N}_+$ appear at least once but can be arbitrarily often and additionally there can be 
        \begin{enumerate}[label=(\arabic*)]
            \item no legs of length $3$ or $4$,
            \item one leg of length $4$ and no leg of length $3$,
            \item one leg of length $3$ and no leg of length $4$.
        \end{enumerate}
    \end{enumerate}
    \end{reptheorem}

    \begin{proof}
        Given any initial graph $\Gamma$ we will prove such a sequence exists by induction on its vertices. It is clear that if we pick any vertex of $\Gamma$ at random, and then add some vertex in its neighbourhood, this trivially is a graph in \textit{A}. Assume now that after having added $N$ vertices one has a graph $\Gamma_N$ of the form of one of our canonical types. We show that $\Gamma_{N+1}=\Gamma_N\cup\{V\}$, for some $V\in\mathcal{N}(\Gamma_N)$, can always be mapped to a canonical graph as well. Take a lightning $\chi=\Gamma_{N+1}(V)$. By \Cref{thm:1_lit} this is always equivalent to a lightning where only one vertex is lit, such that $V$ shares just one edge with $\Gamma_N$. For the case where several legs of length 1 where initially in different lit configurations in $\chi$, the new vertex creates a new leg of length 2. By \Cref{lemma:shortening} this can then be transformed into one of the canonical graphs. If all legs are in the same state, for graphs in \textit{A}, if $O$ or the outermost vertex of the long leg $L$ are lit, then the graph is still in class \textit{A}. Otherwise, if $V$ is connected to some vertex $L_j\in L$ other than the one at its end, then $deg(L_j)=3$. In this case, we can take a lightning $\chi'=\Gamma_{N+1}(L_{j+1})$ and toggle according to 
        \begin{equation}
            \label{eqn:contraction_move_legs}
            L_j,L_{j-1},\dots,L_1,O,\omega,V,L_j,\dots,L_1,O
        \end{equation}
        with $\omega$ one of the legs of length 1 of $O$. This removes the edge $(L_j,L_{j+1})$ and creates the edge $(O,L_{j+1})$. The resulting graph is a star graph with just one central vertex, which by \Cref{lemma:shortening} can be transformed into one of our canonical graphs. For all other cases, namely for $\Gamma_N$ of type $B1,B2,B3$, whenever $V$ is connected to $O$, $\Gamma_{N+1}$ remains in $B1,B2,B3$ respectively. Otherwise, $V$ is connected to some vertex in a leg of length 2, the leg of length 4 or the leg of length 3, respectively. If this is actually the end vertex of these legs, either $\Gamma_{N+1}$ is already of our canonical types, or by \Cref{lemma:shortening} it can be converted to it. Instead, if V connects to some other vertex $L_j$, such that $deg(L_j)=3$ in $\Gamma_{N+1}$, we can reproduce the same argument as before. After taking a lightning $\chi'=\Gamma_{N+1}(L_{j+1})$ and repeating the sequence in \Cref{eqn:contraction_move_legs} we transform $\Gamma_{N+1}$ back into a star graph, which again by \Cref{lemma:shortening} can be transformed into one of our canonical graphs.
    \end{proof}

\section{Elements of graph Lie-algebra}
In order to describe Paulis within the Lie algebra spanned by some generators we introduce the concept of colouring. 
In general, we are interested in the Paulis in the algebra. As such, we can associate to a particular Pauli
\begin{align}
    P_C=\prod_{P_i\in \mathcal{G}} P_i^{C_i}
\end{align}
a vertex colouring $C\in \{0,1\}^{n_G}$ of the graph, such that a vertex $v_i$ is coloured if $C_i=1$. We say that $P_C$ is in the Lie algebra if $C$ is a \emph{valid} colouring. The map $C\mapsto P_C$ does not need to be injective, since the $P_i$ can have algebraic relations between them.

When we consider elements of the Lie-algebra, we require that they can be written as a linear combination of nested commutators of the 
generators
\begin{align}\label{eq:nested_com}
    P=[\cdots [[P_{i_1},P_{i_2}],P_{i_3}]\cdots P_{i_{n_L}}]\,.
\end{align}
We do not need to consider commutators of commutators since by the Jacobi identity
\begin{align}
    [[A,B],[C,D]]=[[[A,B],C],D]]-[[[A,B],D],C]].
\end{align}
In particular, for Paulis, in order for the term in the $lhs$ to be non-zero, one of the two terms in $rhs$ has to vanish, due to the rule that for 3 anti-commuting Paulis, the commutator of 2 of them commutes with the third. This then means that there is always a way to rewrite some non-zero nested commutator into the form of \Cref{eq:nested_com}. As such, we can simply consider how do graph colourings change when taking commutators with the generators. In order for the commutator to be non-trivial, we need that the Pauli anti-commutes with the generator. As such, the rule is that a colour of a vertex can be flipped, if an odd number of its neighbouring vertices are coloured.

Having found the previous fundamental graphs, now we can try to exploit their cycle-less structure to find a set of rules to determine which colourings are \emph{valid}, i.e., which generators give non-zero commutators.

\begin{lemma}[Valid colourings]
    \label{lemma:non_zero_comm}
    For any graph in its canonical form, we have that
    \begin{enumerate}
        \item for a path, any colouring that has one connected component and potentially an additional even numbers of legs of length 1 that are coloured corresponds to valid colouring,
        \item for a star, any colouring with an odd number of connected components is a valid colouring, except for the case with a single leg of length 3 (\textit{B3}), when the first and third vertex of the long leg and an odd number of legs of length 1 are coloured.
    \end{enumerate}
    Here, connected component refers to the connected component of the sub-graph induced by only coloured vertices.
\end{lemma}

\begin{proof}
We first show that every colouring that can be generated by the allowed transformations needs to obey these conditions.
We use that the graphs do not have loops. Therefore, connectivity is locally preserved, meaning that two coloured vertices in the neighbourhood of an uncoloured vertex always belong to two different connected components.
If a vertex has only one coloured neighbour, then flipping it only adds/removes the vertex from the respective connected component, but never changes the number of components. In general flipping a vertex with $2k+1$ coloured neighbours then merges/disconnects $2k+1$ components. As such, the parity of the number of connected components can never change (since flipping a vertex with an even number of coloured neighbours is not allowed). Since by \Cref{eq:nested_com} we always start our colouring from a single vertex, the number of connected components will always be odd. Additionally, for the canonical types, this means that the number of connected components only changes by flipping the central vertex. 
In particular, this implies that for graphs in class \textit{A} we cannot have different connected components along the long leg. Given \Cref{eq:nested_com} one can check, that in order to create several connected components on the long leg, on should start from a colouring with one single connected component spanning over some vertices of the long leg, the center, and potentially some even number of legs of length 1. Uncolouring $O$, would then create an odd number of connected components, one of which on the long leg. Given that colour flips on the long leg, will not create new connected components by themselves, the only way to create more connected components, would be to uncolour the innermost vertex of this leg, and colour back the center. This however, cannot be done, since the center would then just have an even number of coloured neighbours, preventing us from colouring it back.

Finally, for the exception in $B3$, where an odd number of legs of length 1 and alternating vertices of the leg of length 3 are coloured, we note, that this corresponds to a configuration where no additional vertices are allowed to be flipped. As such, this colouring cannot be constructed from some sequence of vertex colourings.

To show that every such colouring can also be generated, we will give an explicit algorithm. 
\begin{enumerate}
    \item Start with a coloured central vertex.
    \item Colour the first vertex of all legs which in the final configuration have at least one coloured vertex. 
    \item Potentially colour a leg of length 1 $\omega$, then uncolour the center.
    \item Move each connected component into the correct position, by colouring vertices adjacent to it. 
    \item If there is a second connected component in the long leg $L$ (3 or 4), first move the outer into position. Then, take some leg $T$ of length 2 with some coloured vertex. Depending on the parity of the number of neighbours of $O$ that are coloured, flip the colours in $T$. Colour the central vertex $O$ and $L_1$. Uncolour the leg of length 1 $\omega$ and uncolour $O$ again.
    
    \item If desired, colour the center back on and flip the leg of length 1 again.
   
\end{enumerate}
This algorithm will prepare any valid graph colouring as long as there is a connected component on a leg of length 2 or no two connected components on the same leg. The two remaining cases are when the long legs $L$ have two connected components and only legs of length 1 are coloured. For the case  \textit{B3} (one leg of length 3), these are not able to be prepared as argued above. 
For \textit{B2} (one leg of length 4), however, we can prepare them. In order to show this, we just need to argue that a colouring where the central vertex and the second and fourth vertex of the long leg are coloured, is a valid colouring. 
Starting from the coloured center we achieve this by flipping the colours according to
\begin{align}
    L_1,L_2,L_3,L_4 ,T_1,\omega,O, T_2,T_1, L_1,L_2,L_3, O,\omega ,L_1,O,L_2,L_1 T_1,T_2,O,T_1
\end{align}
where $L$ refers to the long leg, $T$ to a leg of length 2 and $\omega$ to a leg of length 1. After this, we can colour the rest of legs of length 1, which is an odd number, and uncolour $O$. Finally, on $L$, we can trivially go from this colouring to any other with two connected components.

\end{proof}

\section{Algebraic dependent generators}
\label{sec:appendix_algd}

Like we discussed in the main text, we are not only interested in classifying Pauli Lie algebras, but also in determining which can be mapped to each-other through Clifford operations. To do so, we introduced \Cref{thm:clif_algdeb} to characterize algebraic dependencies between generators, which ultimately will determine whether or not two sets of Paulis are Clifford equivalent. We prove this theorem here.

\begin{reptheorem}{thm:clif_algdeb}[Limits to algebraic dependencies]
        Every minimal generator Pauli Lie algebra has at most one algebraic dependence. Algebraic dependencies can only occur for graphs of type \textit{A} and  \textit{B3}. The two possibilities are shown in \Cref{fig:canonical_types} (c)
        As such, there is a total of $6$ Clifford inequivalent families of Pauli Lie algebras.
\end{reptheorem}

\begin{proof}
    It is enough to prove this just for graphs in the types of \Cref{thm:types}, as contractions preserve algebraic and Lie algebraic dependencies. Then, we will show that algebraic dependencies on these graphs need to involve a colouring where the end vertices of some legs of odd length are colored, have alternating structure, and can only involve one leg of length 1. Overall, this will just leave the two cases mentioned in the statement. 
    
    Assume we have an algebraically dependent subset of $\mathcal{G}$, that is $\{P_i\}_{i=1}^m$ such that $P_1\dots P_m=\mathbf{1}.$ Since $\mathbf{1}$ commutes with everything, the colouring of $\{P_i\}_{i=1}^m$ on the graph, must share an even number of edges with every vertex in it.
    Then, for some leg $L\subseteq\Gamma$ (which for class \textit{A} might be the path itself), if some element in $L$ is coloured, i.e.,  $\{P_i\}_{i=1}^m\cap L\neq\{\emptyset\}$, then the end vertex of $L$ must be coloured. Otherwise, we could find some vertex in $L$ sharing just one edge with the colouring. From this, it is immediate to see that the colouring must have an alternating structure (coloured, then uncoloured). If the end vertex of a leg is coloured and its neighbour as well, then the end vertex itself would share just one edge with the coloured set. The alternating structure then follows from carrying on with this argument. 
    Finally, the center can never be coloured, as otherwise the vertex of a leg of length 1 would only have one coloured neighbour.
    It follows that an algebraic relation cannot have a colouring involving legs of even length, as there is no alternating colouring that colours the outer vertex, but not the center.
    
    If we have a colouring involving more than one leg of length $1$, we can show that this either violates the condition that the generators are Lie-algebraically independent, or that the graph can be mapped to a form where the algebraic relation involves at most one leg of length one. 
    Let $\{L_i\}_{i=1\dots 2k+1}$ be coloured legs of length 1. We can contract $L_2,\dots, L_{2k+1}$ onto $O$, yielding $O'=OL_2\dots L_{2k+1}$. Contracting now $O'$ onto $L_1$, results in $L_1'=L_1L_2\dots L_{2k+1}O$. While this modifies the connectivity of the graph, if we now contract again all $L_2,\dots, L_{2k+1}$ onto $O'$, and then again the new $O''=O$ onto $L_1'$, this undoes the change, and leaves all vertices the same, except for $L_1$ which is mapped to $L_1''=L_1,\dots , L_{2k+1}$. After this, the only coloured leg of length 1 is $L_1''$.
    Since the total number of coloured legs need to be even, this transformation allow us to get to two scenarios. The first is when only two legs of length one are coloured, in which case $P_1P_2=\mathbf{1}\iff P_1=P_2$, which contradicts the Lie algebraic independence. As such, only the second case is relevant, when one leg of length $1$ and a longer leg of odd length are coloured. This is only possible in case $A$, when the length of the long leg is odd and B3.
    One last remark, is that there cannot be several such lightnings involving different legs of length 1, as this would imply an algebraic dependence between the legs of length 1, which we already showed is not allowed.
   
\end{proof}

\section{Controlled Lie-algebras and removing legs of length 1}
\label{app:controls}

In \Cref{thm:graph_to_lie} in the main text, we show that in the presence of several legs of length 1, $2^{n_c}$ direct sum terms appear in the description of the Lie algebra. In this appendix, we show why this happens and what these several legs of length 1 on the graph represent in physical terms.

\begin{lemma}[Controlled Lie-algebras]\label{lem:control}
    Given a graph $\Gamma$ with two legs of length 1 ($P_1,P_2$), its Lie algebra splits into a direct sum of two smaller Lie algebras, in particular those with anti-commutation graph $\Gamma\setminus\{P_2\}$. This also corresponds to a physical direct sum over the qubits upon which $Z_i\coloneqq P_1P_2$ acts.
\end{lemma}

\begin{proof}
    It first should be noted that the product $Z=P_1P_2$ is a symmetry in the Lie algebra, since it shares an even number of edges with every generator. Now, we can use the fact that in order to find the minimal number of qubits for a set of algebraically independent generators, we count the number of disconnected points and pairs in the \textit{algebra} graph.
    Namely, it is easy to see that if we allow ourselves to perform contractions between disconnected points (which one can do if interested just in the algebra, not the Lie algebra) any given graph $\Gamma$ can be transformed into the disjoint union of some pairs of connected vertices and isolated vertices. Then if all generators are independent, we necessarily need to add a qubit for each such connected component, as connected pairs can be mapped to $(X_i,Z_i)$ on some qubit $i$, and single points to $(Z_j)$ on some different register $j$. Hence, since there cannot be an algebraic dependence between legs of length 1 (shown in \Cref{thm:clif_algdeb}), and $P_1P_2$ is a disconnected point, this means that this must necessarily involve adding a new qubit register. This in turn means that we can always find a Clifford operation that maps $Z=P_1P_2\mapsto Z_N$, where $N$ is the new register. 

    Now, given some arbitrary element $P\in\mathfrak{g}(\Gamma)$, we can show that $Z\cdot P\in\mathfrak{g}(\Gamma)$. Using the terms that we introduced in previous sections, we need to show that given a colouring on the graph such that the resulting $P$ is in the Lie algebra, then adding $P_1,P_2$ to this colouring is also in $\mathfrak{g}$. However, from \Cref{lemma:non_zero_comm} this follows directly, since flipping the colouring of $P_1,P_2$ does not change the parity of the number of connected components, and thus also is a valid colouring. This then splits the Lie algebra in
    \begin{equation}
        \{\Bar{P}\otimes\mathbf{1}_N\}\bigoplus\{\Bar{P}\otimes Z_N\}=\mathfrak{g}.
    \end{equation}
    Now, since \textit{wlog} the only term containing $Z_N$ is $P_2$, for every element $\Bar{P}\otimes\mathbf{1}_N\in\mathfrak{g}$, its colouring must be in $\Gamma\setminus\{P_2\}$. On the other hand, for any element $\Bar{P}\otimes Z_N\in\mathfrak{g}$, its colouring must be in $\Gamma\setminus\{P_1\}$, with $\mathfrak{g}(\Gamma\setminus\{P_1\})\cong\mathfrak{g}(\Gamma\setminus\{P_2\})$ because the graphs are the same. 
\end{proof}

Thus, the two takeaways from this are that, one the one hand, we can characterize graphs with just one leg of length 1, as all other Lie algebras will be direct sums of these and, on the other hand, that these direct sums, are effectively also a direct sum in Hilbert space, i.e., they entail a whole new register which is acting as a control on the dynamics. By repeated application of \cref{lem:control}, it follows that the Lie algebra decomposes into $2^{n_c}$ many direct sums in Hilbert space, which are Clifford equivalent to \begin{align}
    Z_1^{i_1}\otimes\cdots \otimes Z_{n_c}^{i_{n_c}} \otimes P_{\mathfrak{g}}
\end{align}
for $\vec i\in \{0,1\}^{n_c}$.
Finally, since a Lie algebra is a vector space an equally valid basis is
\begin{align}
    \ketbra{\vec i}\otimes P_{\mathfrak g}
\end{align}
making explicit the effect of registers $1$ to $n_c.$

\section{Lie algebra classification}

In this section, we will constructively prove the rest of \Cref{thm:graph_to_lie}. As the key ingredient to do so, will be introducing a canonical set of Pauli labels for each graph family. Then, since the Lie algebra type just depends on the graph, we will just need to characterize these canonical sets in order to have a full classification. We show in \Cref{fig:clifford_labeling} all 6 Clifford inequivalent graphs, together with 
their corresponding canonical labels.

\begin{figure}[ht]
    \centering
    \includegraphics[scale=0.25]{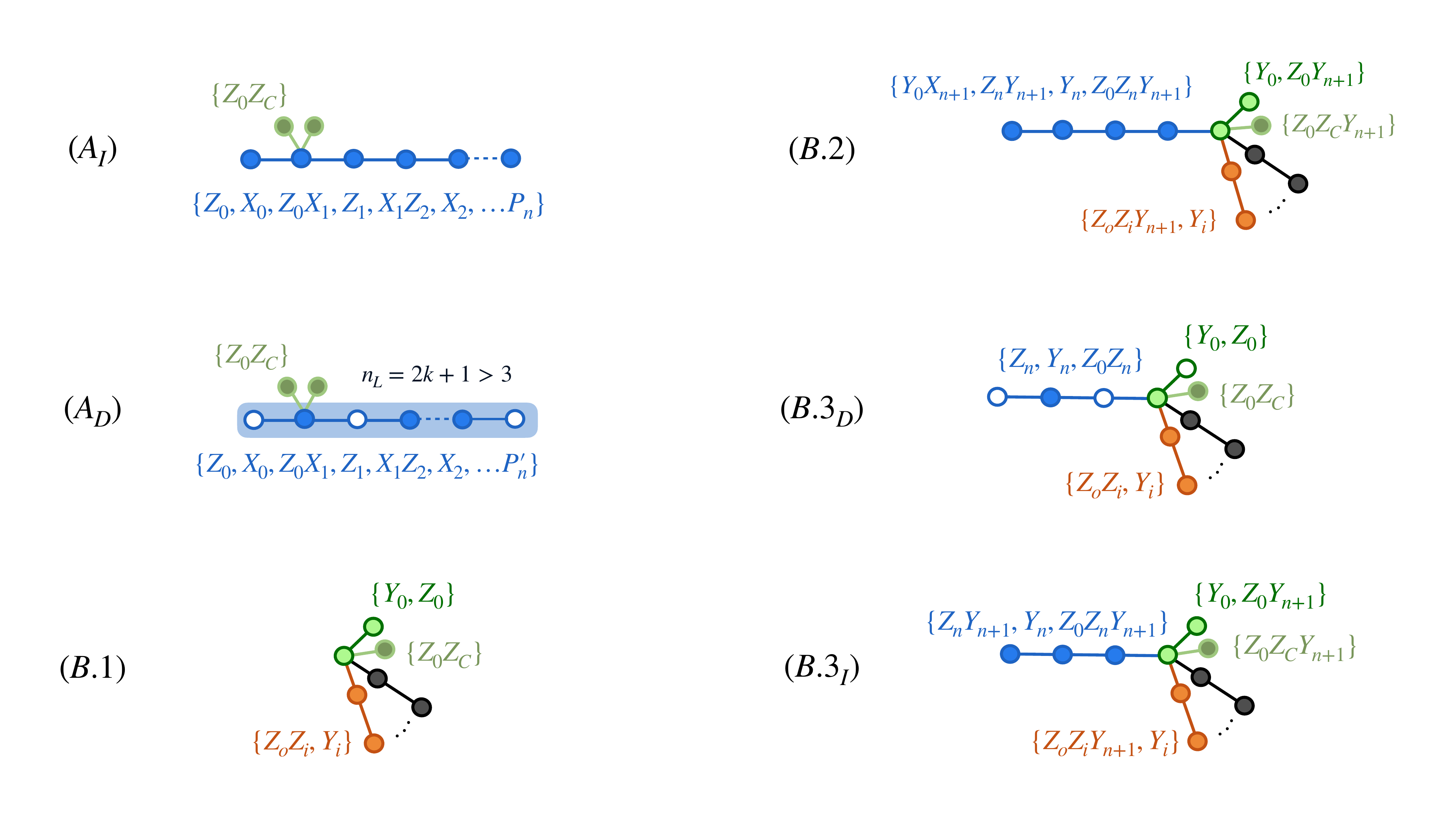}
    \caption{All 6 Clifford inequivalent Lie algebraic classes with their respective canonical labeling as stated in the text. }
    \label{fig:clifford_labeling}
\end{figure}

Let us start with graphs of type \textit{A}, both with algebraic dependence and without. In this case, we simply need to choose a set of closest neighbour interacting Paulis, 
\begin{equation}
    S_{\textit{A}}=\{Z_0,X_0,Z_0X_1,Z_1,X_1Z_2,X_2,\dots P_m\}\cup\{Z_0Z_{C}^i\}_{i= 1,\dots,n_C},
\end{equation}
together with a set of labels for additional control legs,
where $P_m$ depends on the parity of $n_L$ and whether we are in the $(A_D)$ or $(A_I)$ instance. For $n_L=2k$ all terms must be algebraically independent, so that $P_m\in\{Z_{n_L/2-1}X_{n_L/2},X_{n_L/2-1}Z_{n_L/2}\}$ depending on the parity of $n_L/2$. For $n_L=2k+1>3$ we can either have algebraic dependence or not. For $(A_I)$, $P_m\in\{Z_{(n_L-1)/2-1}X_{(n_L-1)/2},X_{(n_L-1)/2-1}Z_{(n_L-1)/2}\}$, whereas in $(A_D)$, $P_m\in \{Z_{(n_L-1)/2},X_{(n_L-1)/2}\}$. With this, we can now prove that indeed this corresponds to the free-fermionic Lie algebra. We can limit to characterizing the Paulis lying on the path graph without legs of length 1, as controls can be removed at this point. From these, it is easy to define a set of $n_L$ majorana operators. In particular, $wlog$ we use Paulis in $(A_I)$ to do so, as 
this will not change the Lie algebra. Thus
\begin{align}
        \gamma_0 &= Z_0,\\
        \gamma_1 &= Y_0,\\
        \gamma_2 &= X_0X_1,\\
        \dots \\
        \gamma_{n_L-1} &= \Pi_{i=0}^{m} 
        P_i,
    \end{align}
    which satisfy
    \begin{equation}
        \comm{\gamma_i}{\gamma_j} = \pm 2i\gamma_i\gamma_j,\hspace{0.5cm}\forall i,j
    \end{equation}
    and thus
    \begin{equation}
        \comm{\gamma_k}{\gamma_i\gamma_j} = 0,\hspace{0.5cm}\forall k\neq i,j.
    \end{equation}
    Hence, a basis for the Lie algebra is given by $\{\gamma_i\}_{0\leq i\leq n_L-1}\cup\{\gamma_i\gamma_j\}_{0\leq i<j\leq n_L-1}$, where quadratic terms satisfy 
    \begin{equation}
        \comm{\gamma_i\gamma_j}{\gamma_r\gamma_s} = \pm 2i(\delta_{j,r}\gamma_i\gamma_s-\delta_{s,i}\gamma_r\gamma_j).
    \end{equation}
    The size of this basis then gives us that \begin{equation}
    \dim(\mathfrak{g})=\frac{n_L(n_L+1)}{2}, 
    \end{equation}
    which coincides with the dimension of $\mathfrak{so}(n_L+1)$. Finally, taking a basis of anti-symmetric matrices of size $(n_L+1)\times(n_L+1)$, it is easy to check that the structure constants of $\mathfrak{g}$ and $\mathfrak{so}(n_L+1)$ are also the same, thus proving equality.
    The difference between $(A_I)$ and $(A_D)$ then comes from the fact that for the latter, the last Majorana 
    operator $\gamma_{n_L-1}$ is equal to the global parity operator.

    For the rest of star-graph families, before giving the labeling, we will first introduce the following lemma regarding properties of certain colourings on the graph. With this lemma, proving the Lie algebra type will follow almost directly.

\begin{lemma}[Odd colouring property]
    \label{lemma:odd_colouring_property}
    If all generators $P\in \mathcal{G}$ of a star-type Pauli Lie algebra satisfy 
    \begin{align}
        \label{eqn:defining_condition}
        P^T Q+QP=0
    \end{align}
    for some Pauli $Q$, then every Pauli $K=\prod_{i\in s\subset \mathcal{V}, |s|\geq 1} P_i$
    it holds that
    \begin{align}
        \textrm{the colouring corresponding to} \, K\,\,\textrm{has an odd number of connected components}\Leftrightarrow K^TQ+QK=0.
    \end{align}
\end{lemma}

\begin{proof}
    For the set of $S=\{\prod_{i\in s} P_i| s\subset \mathcal{V}, s\neq \emptyset\}$ of all Paulis generated by some colouring,
    we define the two subsets 
    \begin{align}
        S_i =\{P\in S | P^TQ+(-1)^i QP=0\}, \quad \text{for } i \in \{0,1\}.
    \end{align}
    In general, we have for $L=V\cdot W$, where $V\in S_{i_V}$ and $W\in S_{i_W}$
    \begin{align}
    \nonumber
        L^TQ\pm QL&=(VW)^TQ\pm QVW=W^TV^T Q\pm QVW
        =-(-1)^{i_V} W^TQV\pm QVW\\
        &=(-1)^2(-1)^{i_V+i_W}QWV \pm QVW=Q((-1)^{i_V+i_W}WV \pm VW).
        \label{eq:transpose_commutator}
    \end{align}
    As a consequence, we have that if $V,W\in S_0$ and $\{V,W\}=0$, then $L=VW\in S_0$.
    This means that if we have a connected component $C_i\subset \mathcal{V}$, it follows that $P_{C_i}\coloneqq\prod_{j\in C_i} P_j\in S_0$. This is because we can generate this by a nested sequence of generators. Since the graph has no loop, we can start at any element, and then take a commutator with a Pauli from $C_i$ that anti-commutes, until all Paulis were used exactly once. 
    Now we consider the disconnected components. We have from \Cref{eq:transpose_commutator} that 
    $P_{C_1}P_{C_2}\in S_1$ since $[P_{C_i},P_{C_j}]=0$. By iterative application we have that 
    \begin{align}
        K=P_{C_1}\dots P_{C_{n_{cc}}}\in S_{n_{cc}+1\mod 2}.
    \end{align}
    This means that the Pauli $K$ satisfies the condition if the number of connected components $n_{cc}$ is odd, and satisfies $K^TQ- QK=0$ if $n_{cc}$ is even. This gives the if and only if relationship that is desired by the lemma.
\end{proof}
\label{app:graph_lie_proof}

Having stated this, we can start discussing graphs $\textit{B1}$. In this case, there is just one Clifford inequivalent class and we choose the labeling to be given by
\begin{equation}
    S_{\textit{B1}}=\{Y_0,Z_0\}\cup\{Z_0Z_i,Y_i\}_{i=1,\dots , n_2}\cup\{Z_0Z_{C}^i\}_{i=1,\dots,n_C}
\end{equation}
where these sets label the central vertex and the first leg of length 1, all legs of length 2, and controls, respectively. Note that the central Pauli $Y_0$, is chosen so that $Q-\Omega$ as in the canonical definition in \Cref{app:preliminaries}  in the definition of $\mathfrak{sp}$. Moreover, every vertex adjacent to the center is labeled by a symmetric Pauli $Z_0,Z_0Z_i$, whereas non-adjacent vertices are labeled with an anti-symmetric Pauli $Y_i$. This precisely matches the defining condition for $\mathfrak{sp}$, i.e.,
\begin{equation}
    P^TY_0=-Y_0P.
\end{equation}
Hence, since all our generators satisfy this condition, which is preserved under Lie product, all elements in our Lie algebra will, i.e.
\begin{equation}
    \operatorname{span}(\langle S_{\textit{B1}}\rangle)\subseteq \bigoplus_{i=1}^{2^{n_c}}\mathfrak{sp}(2^{n_2}).
\end{equation}
For us to have equality as stated by \Cref{thm:graph_to_lie}, one would have to check that any element in $\bigoplus_{i=1}^{2^{n_c}}\mathfrak{sp}(2^{n_2})$ can be generated by our Paulis. This, however, follows from \Cref{lemma:odd_colouring_property}. Given some $K\in\mathfrak{sp}(2^{n_2})$, this can be written as a product of our generators, since these form a basis for the algebra of symplectic matrices of size $(n_2+1)\times(n_2+1)$. This follows from the same arguments as for finding the minimal number of qubits. Since $2n$ independent Paulis are required to span the full algebra on $n$ qubits, in order to span $\mathfrak{sp}$ one also needs at most $2n$ symplectic generators, which is what we have. Thus, there is some colouring on the graph corresponding to $K$. Since $K$ satisfies
\begin{equation}
    K^TQ+QK=0
\end{equation}
with $Q=Y_0$, then the colouring of $K$ must have an odd number of connected components. Now by \Cref{lemma:non_zero_comm} we know that such colourings can indeed be computed through nested commutators, thus concluding the proof that the Lie algebra spanned by this set of Paulis is indeed $\bigoplus_{i=1}^{2^{n_c}}\mathfrak{sp}(2^{n_2})$.

Given this labeling for the symplectic Lie algebra, one can easily find one for the family of graphs in \textit{B3}, both for $\textit{B3}_I$ and $\textit{B3}_D$. For the case of algebraic dependent generators, we can simply get our canonical labeling by adding $Z_n$, for $n\in\{1,\dots,n_2\}$.  It is easy to see that if one takes the product of Paulis in the leg of length 1 and the innermost and outermost vertices of the newly created leg of length 3
\begin{equation}
    Z_0\cdot Z_0Z_i\cdot Z_i=\mathbf{1}.
\end{equation}
Adding this additional Pauli allows us to generate any Pauli on $n_2+1$ qubits, i.e., promotes the Lie algebra to $\mathfrak{su}(2^{n_2+1})$. That is because given a Pauli on $n_2+1$ qubits, we can always find a colouring which is computable, since in the case where $P$ has a colouring with an even number of connected components, then $Z_0\cdot Z_0Z_i\cdot Z_i\cdot P$ is a colouring with an odd number of connected components that yields the same Pauli.

If the generators are not algebraically dependent, by the rules that determine the minimal number of qubits, we must add a new qubit as well. Hence, instead of adding $Z_n$, we will introduce $Z_nY_{n+1}$. Moreover, since every other Pauli commutes with $Y_{n+1}$, we can freely modify them so that the overall labeling for this case is
\begin{equation}
    \label{eqn:su_ind}
    S_{\textit{B3}_I}=\{Y_0,Z_0Y_{n+1}\}\cup\{Z_0Z_iY_{n+1},Y_i\}_{i=1,\dots , n_2-1}\cup\{Z_0Z_{C}^iY_{n+1}\}_{i=1,\dots,n_C}\cup\{Z_0Z_nY_{n+1},Y_n,Z_nY_{n+1}\},
\end{equation}
 it is clear that now 
\begin{equation}
    Z_0Y_{n+1}\cdot Z_0Z_iY_{n+1}\cdot Z_iY_{n+1}=Y_{n+1}
\end{equation}
so we cannot use the same argument as before. However, it suffices to prove it for one case, since the graph is the same for both. Furthermore, the last change was useful to see that as we stated in the main text, the map between $\textit{B3}_D$ and $\textit{B3}_I$ is given by $P\mapsto i\mathrm{Im}(P)\otimes \mathbf{1}_{n+1}+\mathrm{Re}(P)\otimes Y_{n+1}$, so that indeed the Lie algebra is the same in both cases, albeit in $\textit{B3}_I$ it is embedded in a bigger Hilbert space.

The other reason why appending $Y_{n+1}$ to symmetric generators was because we can now easily find our canonical labeling for the last class of graphs $\textit{B2}.$ Since we want to argue that this corresponds to the Lie algebra of anti-symmetric matrices $\mathfrak{so}$, and all generators in \Cref{eqn:su_ind} are already anti-symmetric, we just need to add another anti-symmetric Pauli that breaks the symmetry with respect to $Y_{n+1}$. To that end, our new set of canonical generators will be given by 
\begin{equation}
    S_{\textit{B2}}=\{Y_0,Z_0Y_{n+1}\}\cup\{Z_0Z_iY_{n+1},Y_i\}_{i=1,\dots , n_2-1}\cup\{Z_0Z_{C_i}Y_{n+1}\}_{C_i\in I}\cup\{Z_0Z_nY_{n+1},Y_n,Z_nY_{n+1},Y_0X_{n+1}\}.
\end{equation}
Then, since all generators satisfy the anti-symmetric condition, which is preserved under Lie product, we have
\begin{equation}
    \operatorname{span}(\langle S_{\textit{B2}}\rangle)\subseteq \bigoplus_{i=1}^{2^{n_c}}\mathfrak{so}(2^{n_2+3}).
\end{equation}
Like for the symplectic case, checking that any element in $\bigoplus_{i=1}^{2^{n_c}}\mathfrak{so}(2^{n_2+3})$ can be computed through nested commutators of our canonical generators is easy to see from \Cref{lemma:odd_colouring_property} with $Q=\mathbf{1}$. Again, for any Pauli $K$ satisfying 
\begin{equation}
    K^T\mathbf{1}_{n_2+3}+\mathbf{1}_{n_2+3}K=0
\end{equation}
its colouring on the graph must have an odd number of connected components, and thus must be computable through nested commutators. 
This concludes our construction of all canonical Pauli sets for every Clifford inequivalent class, and serves as a proof of \Cref{thm:graph_to_lie}.

\subsection{Clifford invariant form}

In the previous section we mainly worked with the canonical definition for $\mathfrak{sp}$ and $\mathfrak{so}$. However, since Clifford operations do not change the Lie algebra, one could come up with a form that makes this explicit.
Before, we introduced the canonical forms as
\begin{align}
             \mathfrak g=\mathrm{span}(\{ P\in \mathcal{P}| P^T Q+QP=0 \})\,,
\end{align}
where $Q=\mathbf{1}$ for $\mathfrak{so}$ and $Q=\mathbf{Y_0}$
for $\mathfrak{sp}$.
If we perform a Clifford transformation, this turns into 
\begin{align}
    P\mapsto P'=C^\dagger P C.
\end{align}
As such, the statement becomes 
\begin{alignat}{3}
    &&0&=(C P' C^\dagger)^T Q+Q (C P' C^\dagger)\\
    &&&=C^* P^{'T} C^TQ+Q C P' C^\dagger\\
    \Rightarrow && 0&=P^{'T} C^T QC+C^TQ C P',
\end{alignat}
using $C^\dagger C=C^TC^*=\mathbf{1}$.
With this the action maps
\begin{align}
    Q\mapsto Q'= C^T QC.
\end{align}
This map still maps Paulis to Paulis upon conjugation up to phase. To see this, we can look at a particular generating set of the Clifford group, e.g.,  $(S,H,CNOT)$. On the one hand, since $H^T=H^\dagger$ and $CNOT^T=CNOT^\dagger$, these act as ordinary Cliffords, mapping Paulis to Paulis. On the other hand, for the map $Q\mapsto  S^{T}QS$, we have
\begin{align}
    \mathbf{1}\mapsto Z,\quad X\mapsto i X,\quad Y\mapsto iY,\quad Z\mapsto \mathbf{1}.
\end{align}
However, for the relation that concerns us, this global phase becomes irrelevant. 
Additionally, we have that
\begin{align}
    Q'^T=C^T Q^TC = Q'\times \begin{cases}
        1 &Q^T=Q,\\
        -1 & Q^T=-Q,
    \end{cases}
\end{align}
i.e., the symmetric/anti-symmetric property is preserved under the 
action of this map. Additionally, we can map any (anti-)symmetric Pauli to any other (anti-)symmetric Pauli. For this, we can change $\mathbf{1}\mapsto Z$ by applying $S$, $Z\mapsto X$ by applying $H$ and  $X\mapsto I$, by applying $HS$. In order to change the number of $Y$ Paulis we can use CNOTs to map $Y\otimes Y\mapsto X\otimes Z$. 
This then motivates the definition that is given in the main text.

Another interesting case in this regard, is that of the class $\textit{B3}_I$ where the canonical form is defined by
\begin{align}
             \mathfrak g=\mathrm{span}(\{ P\in \mathcal{P}\backslash \{\mathbf{1}\}|[Y_{n+1},P]=0 ,\,P^T+ P=0 \})\,.
\end{align}
Again, applying a Clifford to the argument this becomes 
\begin{align}
    0&=(C P' C^\dagger)Y-Y(C P' C^\dagger), \quad 0=(C P' C^\dagger)^T +(C P' C^\dagger),\\
    0&= P' C^\dagger Y C-C^\dagger YC P',  \quad 0=P^{'T} C^T C+C^T C P',
\end{align}
meaning 
\begin{align}
    K\coloneqq C^TC\, ,\quad Q\coloneqq C^\dagger YC.
\end{align}
Thus, when the condition is satisfied
\begin{align}
    Q^TK+KQ&=(C^\dagger YC)^T (C^TC)+(C^TC)(C^\dagger YC)\\
    \nonumber
    &=C^T Y^TC^*C^TC+C^TCC^\dagger YC\\
    \nonumber
    &=-C^T YC+C^T YC=0.
    \nonumber
\end{align}
Naturally, $Q$ can become every non identity Pauli. Using only $H,CNOT$, it is possible to create every symmetric Pauli: $C=CNOT_{i,n}$ to set it to $X_n$, $H_n$ to set $X_n$ to $Z_n$, and CNOT on $X\otimes Z$ to create $Y\otimes Y$. Moreover, since all gates are real $C^T=C^\dagger$, which means that $\mathbf{1} \mapsto C_1^T C_1 =\mathbf{1}$ as well.

Next, we choose a $C_2$, such that $K=C_2^TC_2$. Naturally, $C=C_1C_2$ maps $K\mapsto K$, $Q\mapsto C_2^\dagger C_1^\dagger YC_1C_2$. Now, choosing $C_1$ such that $C_1^\dagger YC_1=C_2 Q C_2^\dagger$ we get to the result we wanted. However, by design $(C_1^\dagger YC_1)^T=-C_1^\dagger YC_1$, meaning 
\begin{align}
    (C_2 Q C_2^\dagger)^T&=-C_2 Q C_2^\dagger,\\
    C_2^* Q^T C_2^T&=-C_2 Q C_2^\dagger,\\
     Q^T C_2^TC_2&=-C_2^TC_2 Q,\\
     Q^T K&=-K Q\,.
\end{align}
The current reduction requires $K=K^T$. However, we can rewrite the commutation relation in an equivalent form
\begin{align}
    0&=[Q,P]=QP-PQ=QP+KP^TKQ, \\
    0&=KQP+P^TKQ, 
\end{align}
as well as
\begin{align}
    0&=Q^TK+KQ=(KQ)^T+KQ
\end{align}
using that $K=K^T$.
So if we define $L\propto KQ$, we get that $L^T=-L$ and the conditions that $\forall P\in \mathcal{P}\backslash L$:
\begin{align}
    P^TL+LP&=0\, ,\quad P^TK+KP=0\,,
\end{align}
which shows that $B3_I$ corresponds to a system with both a symplectic and an orthogonal relation.
If we now use the symplectic relation we get
\begin{align}
    P^T&=-LPL, \\
    0&=-LPLK+KP, \\
    0&=[P,LK].
\end{align}
As such, defining $Q'\propto LK$, $K'=L$
we have
\begin{align}
    K^T&=K, \\
    (K'Q')^T&=K'Q', \\
    0&=(Q')^T K'+ K' Q'\,,
\end{align}
using that $(K')^T=-K'$.
This means that the requirement for $K$ to be symmetric can be removed, by the outlined transformations. This concludes the condition of the Clifford invariant form given in the main text. 

\newpage

\end{document}